\documentclass[11pt,twoside]{atmp}

\usepackage{graphics,epsfig,psfrag}
\usepackage{subfigure}
\usepackage{amsmath,amssymb,verbatim,amsthm}

\usepackage[all]{xy}
\usepackage{color}

\newfont{\bbold}{msbm9 scaled \magstep1}
\newcommand{\bra}[1]{\langle #1|}
\newcommand{\ket}[1]{|#1\rangle}

\newcommand{\im}{\textrm{Im}\ }
\newcommand{\beq}{\begin{eqnarray}}
\newcommand{\eeq}{\end{eqnarray}}
\newcommand{\zs}{\mbox{\bbold Z}}
\newcommand{\R}{\mathcal R}

\newcounter{thm}
\newcounter{lem}
\newcounter{ex}
\newcounter{def}

\newtheorem{theorem}[thm]{Theorem}  
\newtheorem{definition}[def]{Definition}
\newtheorem{lemma}[lem]{Lemma}
\newtheorem{example}[ex]{Example}

\begin{document}

\title[]
{Supersymmetry, lattice fermions, independence complexes and cohomology theory}


\author[L. Huijse \and K. Schoutens]{Liza Huijse \and Kareljan Schoutens}

\address{Institute for Theoretical Physics, University of Amsterdam,
Valckenierstraat 65, 1018 XE Amsterdam, the Netherlands}  
\addressemail{L.Huijse@uva.nl}

\begin{abstract}
We analyze the quantum ground state structure of a
specific model of itinerant, strongly interacting lattice
fermions. The interactions are tuned to make the model
supersymmetric. Due to this, quantum ground states are in
one-to-one correspondence with cohomology classes of the so-called
independence complex of the lattice. Our main result is a complete
description of the cohomology, and thereby of the quantum ground
states, for a two-dimensional square lattice with periodic
boundary conditions. Our work builds on results by J.~Jonsson, who
determined the Euler characteristic (Witten index) via a
correspondence with rhombus tilings of the plane. We prove a
theorem, first conjectured by P.~Fendley, which relates dimensions
of the cohomology at grade $n$ to the number of rhombus tilings
with $n$ rhombi.
\end{abstract}

\maketitle

\cutpage 

\noindent

\section{Introduction}

The motivation for the work presented in this paper is
multi-faceted. On the Physics side, the motivation stems from the
need to understand the electronic properties of materials where
electrons are free to move but subject to strong (repulsive)
interactions. Even at the level of relatively simple model
Hamiltonians, the behavior of such systems is notoriously
difficult to analyze. The model studied in this paper has been
chosen such that it enjoys a property called supersymmetry
\cite{FSdB}. The benefit of this has turned out to be twofold.
First, the supersymmetry leads to a considerable degree of
analytic control, allowing the rigorous derivation of quite a few
results, in particular on quantum ground states. Second, the
sypersymmetric model turns out to have remarkable properties, both
in dimension $D=1$, where the model is quantum critical and
described by a superconformal field theory \cite{FSdB,FNS,BDA}, and in
dimension $D=2$, where the model displays extensive ground state
entropy \cite{FS,vE,HS} and where indications of quantum critical
behavior were found \cite{HHFS}.

On the Mathematics side the study of supersymmetric lattice models has led
to interesting results on the cohomology of independence complexes of
lattices and graphs, 2D grids in particular
\cite{Jonsson1,Jonsson2,Jonsson3,Bax,BM,eng,Csor}. The two sides are connected
by the observation that quantum ground states of the supersymmetric lattice
model are in 1-to-1 correspondence with the elements of the cohomology of an associated
independence complex.

The supersymmetric model on a 2D square lattice has turned out to be particularly
interesting. Numerical results for the Witten index (Euler characteristic) in torus geometry
led to remarkable conjectures for the dependence of this quantity on the two
periods of the torus \cite{FSvE}. These conjectures were then proven by J.~Jonsson
\cite{Jonsson1}, using a connection with specific rhombus tilings of the plane.
In the present paper we complete the analysis by providing a direct characterization of the elements of the
cohomology and thereby of the quantum ground states for the square
lattice wrapped around a torus. We prove a theorem, first conjectured by P.~Fendley,
which relates dimensions of the cohomology at grade $n$ to the number of rhombus
tilings with $n$ rhombi. Since the number of rhombus tilings grows exponentially
with the linear dimensions of the system, our result implies that the quantum model
has a sub-extensive ground state entropy.

The presentation in this paper is organized as follows. In section \ref{physintro}, we introduce
the model and briefly summarize the main results presented in the literature so far. The focus in
this section will be on the physics of the model. We then turn to the mathematics side in section
\ref{mathintro}. We relate the Hilbert space of the supersymmetric lattice model to an independence complex
and show that the quantum ground states of the model are in 1-to-1 correspondence with the elements of the cohomology of
this complex. At the end of this section, we resume the 'tic-tac-toe' lemma of \cite{botttu}, which plays a
central role in the rest of the paper. In section \ref{mainres}, we state the main result of this paper; a
theorem that relates the dimensions of the cohomology at grade $n$ to the number of rhombus
tilings with $n$ rhombi. We briefly discuss how this theorem relates to Jonsson's work
\cite{Jonsson1,Jonsson2} and what the implications for the physics of the model are. The remainder of the
paper (section \ref{proof}) is dedicated to the proof of this theorem. Unfortunately, the proof is quite
involved and consists of several steps. A detailed outline of these steps can be found at the start of
section \ref{proof}.

\section{Physics connection: supersymmetry and lattice fermions}\label{physintro}
In this section we will introduce the model and briefly state the main results obtained for this model with
a focus on the physics interpretation.
\subsection{Supersymmetry}
An $\mathcal{N}=2$ supersymmetric quantum mechanical theory is constructed from a basic algebra, defined by two nilpotent
supercharges $Q$ and $Q^{\dag}$ (complex conjugation is implied) \cite{Witten},
\begin{eqnarray}\nonumber
\{Q,Q\}=\{Q^{\dag},Q^{\dag}\}=0
\end{eqnarray}
and the Hamiltonian given by
\begin{eqnarray}\nonumber
H=\{Q^{\dag},Q\} .
\end{eqnarray}
It satisfies
\begin{eqnarray}\nonumber
[H,Q]=[H,Q^{\dag}]=0.
\end{eqnarray}
The eigenvalues and eigenvectors of the Hamiltonian give the energy spectrum and the corresponding quantum
states. The definition of the Hamiltonian has
some immediate consequences for the energy spectrum. First of all, it is positive definite:
\begin{eqnarray*}
\bra{\psi}H\ket{\psi}
  &=& \bra{\psi}(Q^{\dag}Q+QQ^{\dag})\ket{\psi}
\nonumber \\[2mm]
  &=&|Q\ket{\psi}|^2+|Q^{\dag}\ket{\psi}|^2 \geq 0 \
\end{eqnarray*}
for all choices of the quantum state $\ket{\psi}$.
Second of all, the fact that both $Q$ and $Q^{\dag}$ commute with the Hamiltonian, gives rise to a twofold
degeneracy in the energy spectrum. In other words, all eigenstates of the Hamiltonian with an energy $E_s>0$
form doublet representations of the supersymmetry algebra. A doublet consists of two
states, $\ket{s}$ and $Q \ket{s}$, such that $Q^{\dag}\ket{s}=0$. The states $\ket{s}$ and $Q \ket{s}$ are said to
be superpartners. Finally, all states with zero energy must be singlets:
$Q \ket{g}=Q^{\dag}\ket{g}=0$ and conversely, all singlets must be zero
energy states \cite{Witten}. In addition to supersymmetry our models will also have a particle-number
symmetry generated by the operator $F$ with
\begin{eqnarray}\label{Fnum}
[F,Q^{\dag}]=-Q^{\dag} \quad \textrm{and} \quad [F,Q]=Q.
\end{eqnarray}
Consequently, $F$ commutes with the Hamiltonian. Furthermore, this tells us that superpartners differ in
their fermion number by one (let $F\ket{s}=f_s\ket{s}$, then $F(Q\ket{s})=Q(F+1)\ket{s}=(f_s+1)(Q\ket{s})$).

An important issue is whether or not
supersymmetric ground states at zero energy occur, that is, whether there are singlet representations of the
algebra. For this one considers the Witten index
\begin{equation}
W = \hbox{tr}\left[(-1)^F e^{-\beta H}\right] \ , \label{Windex}
\end{equation}
where the trace is over the entire Hilbert space.
Remember that all excited states come in doublets with the same energy and
differing in their fermion-number by one. This means that in the trace all
contributions of excited states will cancel pairwise, and that the only
states contributing are the zero energy ground states. We can thus evaluate
$W$ in the limit of $\beta \rightarrow 0$, where all states contribute
$(-1)^F$. It also follows that $|W|$ is a lower bound to the number of
zero energy ground states.

\subsection{Lattice fermions}
We now make the model more concrete by defining the supercharges in terms of lattice particles. The particles
we will consider are spin-less electrons, also called spin-less fermions. Their key property is that the
wavefunction is antisymmetric under the exchange of two fermions. It follows that the operator $c^{\dag}_i$
that creates a fermion on site $i$ in the lattice and the operator $c_j$ that annihilates a fermion on site
$j$ in the lattice, satisfy the following anti-commutation relations:
\begin{eqnarray}
\{c_i^{\dag},c_j\}=\delta_{ij} \nonumber\\
\{c_i,c_j\}=\{c_i^{\dag},c_j^{\dag}\}=0. \nonumber
\end{eqnarray}
The particle-number operator for fermions is defined as $F=\sum_i c^{\dag}_i
c_i$, where the sum is over all lattice sites. This operator counts the total number of particles in a state.
A simple choice for the supercharges would be $Q=\sum_i c_i^{\dag}$ and $Q^{\dag}=\sum_i c_i$. It is readily
verified that both obey the nilpotency condition and that the commutation relations with $F$ (\ref{Fnum}) are
satisfied. However, this
choice leads to a trivial Hamiltonian: $H=L$, where $L$ is the total number of sites of the lattice. To
obtain a non-trivial Hamiltonian, we dress the fermion with a projection operator:
$P_{<i>}=\prod_{j \textrm{ next to } i} (1-c_j^{\dag}c_j)$, which requires
all sites adjacent to site $i$ to be empty. We can now formulate the supercharges in terms of these hard-core
fermions: $Q=\sum c_i^{\dag} P_{<i>}$ and $Q^{\dag}=\sum c_i P_{<i>}$. Again the nilpotency condition and the
commutation relations (\ref{Fnum}) are satisfied, but now the Hamiltonian of these hard-core fermions reads
\begin{eqnarray}\label{Hsusygen}\nonumber
H=\{Q^{\dag},Q\}=
\sum_i \sum_{j\textrm{ next to }i} P_{<i>} c_i^{\dag} c_j P_{<j>}
       + \sum_i P_{<i>}.
\end{eqnarray}
The first term is a nearest neighbor hopping term, that is, the fermions can hop from site $j$ to site $i$ as
long as $i$ and $j$ are connected by an edge and provided that the neighboring sites are empty. The
second term contains a next-nearest neighbor repulsion, a chemical potential and a constant. The details of
the latter terms will depend on the lattice we choose.

\subsection{Results and physics interpretation}
This lattice Hamiltonian constitutes a particular instance of an
itinerant-fermion system, where all the interactions are
fine-tuned by the supersymmetry. Over the last few decades,
numerous studies of itinerant-fermion systems in two spatial
dimensions have been presented, however, exact solutions are few
and far between. The model presented here does allow for exact
results and turns out to exhibit quite remarkable features. First
of all, the supersymmetric model on the chain can be solved
exactly through a Bethe ansatz \cite{FSdB}. In the continuum limit
one can derive the thermodynamic Bethe ansatz equations. The model
has the same thermodynamic equations as the XXZ chain at
$\Delta=-1/2$, so the two models coincide (the mapping can be
found in \cite{FNS}). The continuum limit is described by the
simplest field theory with $\mathcal{N}=(2,2)$ superconformal
symmetry with central charge $c=1$, which implies that the model
is quantum critical. For a periodic chain with length $L=3n$ the
model has a twofold degenerate zero-energy ground state with $f=n$
fermions.

This ground state degeneracy turns out to be a generic feature of the model. In fact, in two spatial
dimensions the ground state entropy $S_{\textrm{GS}}$ typically grows exponentially with the system size.
This characteristic of having an extensive ground state entropy goes under the name of superfrustration
\cite{FS}. Numerical studies of the Witten
index have shown that even this lower bound is typically extensive \cite{vE}. Exact results for the number of
ground states were obtained for various lattices \cite{FS}. For example, for the martini lattice, which is formed by
replacing every other site on a hexagonal lattice with a triangle, the number of zero-energy ground states
$e^{S_{\textrm{GS}}}$ was found to equal the number of dimer coverings
of the hexagonal lattice. For large systems ($L\rightarrow \infty$) this gives
\begin{equation}
\frac{S_{\textrm{GS}}}{L}=\frac{1}{\pi} \int_0^{\pi/3} d\theta
     \ln[2 \cos \theta] = 0.16153 \dots \nonumber
\end{equation}
A heuristic way of understanding the superfrustration is from the ``3-rule'': to
minimize the energy, fermions prefer to be mostly 3 sites apart (with details
depending on the lattice). For generic two dimensional lattices the 3-rule can be satisfied in
an exponential number of ways.

For certain two dimensional lattices it was proven that zero-energy ground states exist at various fillings.
The filling is defined as the number of particles per lattice site. For the square, triangular and hexagonal
lattice there exist zero-energy ground states for all rational fillings $\nu$ within the range $[1/5,1/4]$,
$[1/7,1/5]$ and $[1/4,5/18]$, respectively \cite{Jonsson2}.

The square lattice turns out to be a special case. Here the Witten index is subextensive and
for periodic boundary conditions in two directions (i.e. the square lattice wrapped around the torus) it
grows exponentially with the linear dimensions of the system \cite{FSvE,Jonsson1}. In this paper we prove
that the total number of ground states also grows exponentially with the linear dimensions of the system.
In fact, this proof establishes a direct relation between ground states and tilings of the plane with two
types of rhombi (see section \ref{mainres} for details), which were first introduced by Jonsson.
In \cite{HHFS} we presented numerical studies of various ladder realizations of the square lattice with doubly
periodic boundary conditions. These studies, together with the correspondence between ground states and
tilings, strongly indicate the existence of critical edge modes in these systems. It is compelling to
infer that, if a subextensive systems accommodates edge criticality, a truly extensive system, like the
triangular lattice, will allow for bulk criticality. On another speculative note, the ground state-tiling
correspondence for the square lattice on the torus suggests the possible existence of topological order in
some of the many ground states.

\section{Math connection: independence complex and cohomology theory}\label{mathintro}
In this section we will establish the relation between particle
configurations of hard-core fermions and independent sets, on the
one hand, and between zero energy ground states and cohomology
elements on the other. At the end of this section we state the
'tic-tac-toe' lemma for double complexes, which plays a central
role in the proof of the main result of this paper.
\subsection{Independence complex}
An independent set on a graph is a subset of the vertex set of the graph with the property that no
two vertices are adjacent. Since hard-core fermions cannot occupy adjacent sites, it is clear that each
allowed configuration of hard-core fermions forms an independent set. In the following we will use the term
lattice (i.e. a grid) instead of the more general term graph, since in the physics context it is most natural to study
fermions on a lattice. However, the correspondences we establish in this section hold for graphs in general.
The family of independent sets of a lattice forms the independence complex $\Sigma$ of the lattice. We can
define the partition sum in the
asymptotic (thermodynamic) limit for the independence complex $\Sigma$ as
\beq\label{Zstatmech}
Z(\Sigma,z) \equiv \sum_{\sigma \in \Sigma} z^{|\sigma|},
\eeq
where $z$ is called the activity. Until recently there were essentially no exact results for independence
complexes on two dimensional lattices, with one important exception. Baxter \cite{Baxter} gave an analytic
expression for the
partition sum of the independence complex on the triangular lattice with positive activity in the thermodynamic
limit. This is also referred to as the exact solution for hard hexagons (hard-core fermions on the triangular
lattice can, in this context, be viewed as hexagons that share, at most, a side or a corner).

Now observe that the coefficient of $z^k$ in (\ref{Zstatmech}) is the number of sets in $\Sigma$ with size
$k$ or, in other words,
the number of configurations with $k$ hard-core fermions. Consequently, $Z(\Sigma,1)$ gives
the dimension of the full Hilbert space $\mathcal{H}$; the space spanned by all possible hard-core fermion
configurations. What is even more interesting, however, is that $Z(\Sigma,-1)$ coincides with the Witten
index (\ref{Windex})
\beq\nonumber
Z(\Sigma,-1) = \sum_{\sigma \in \Sigma} (-1)^{|\sigma|}= \hbox{tr} (-1)^F.
\eeq

Recently, Jonsson expressed precisely this quantity ($Z(\Sigma,-1)$) for hard squares on a torus, i.e. with
doubly periodic boundary conditions, in terms of rhombus tilings on the torus \cite{Jonsson1} (see section
\ref{mainres} for details). This quantity
coincides with the Witten index for hard-core fermions on the square lattice. For the square lattice the
condition that two particles cannot occupy two adjacent sites readily translates to the hard square condition
if we define the squares to be tilted by $45^{\circ}$ and to have a particle at their center. It follows
that the squares cannot overlap, however they can have a corner or a side in common.

\subsection{Cohomology and homology theory}
It should be clear from the previous that the Hilbert space is a graded vector space, where the
grading is defined by the particle-number operator $F$. That is, the Hilbert space can be written as
$\mathcal{H}=\oplus C_n$, where $C_n$ is a subspace spanned by all the possible configurations with $n$
particles. From the definitions of $F$ and $Q$ and their commutation relations (\ref{Fnum}) it is clear that
$Q$ is a map from $C_n$ to $C_{n+1}$. Since $Q$ squares to zero, we can define its cohomology. On the other
hand, $Q^{\dag}$ is a map from $C_n$ to $C_{n-1}$ and also nilpotent, so we can define the homology of
$Q^{\dag}$.
\beq\nonumber
C_0 \left. \begin{array}{c} Q \\
\longrightarrow \\
\longleftarrow \\
Q^{\dag}
\end{array} \right. C_1 \left. \begin{array}{c} Q \\
\longrightarrow \\
\longleftarrow \\
Q^{\dag}
\end{array} \right. C_2  \left. \begin{array}{c} Q \\
\longrightarrow \\
\longleftarrow \\
Q^{\dag}
\end{array} \right. C_3 \dots
\eeq
It turns out that the zero energy ground states of the model are in one-to-one correspondence with the
non-trivial classes of the cohomology of $Q$ and the homology of $Q^{\dag}$. Remember that all states with
zero energy must be singlets: $Q \ket{g}=Q^{\dag}\ket{g}=0$ and conversely, all singlets must be zero
energy states. Clearly, all singlets, and thus all (zero energy) ground states, are in
the kernel of $Q$:  $Q\ket{g}=0$ and not in the image of $Q$, because if we could
write $\ket{g}=Q\ket{f}$, then $(\ket{f}, \ket{g})$, would be a doublet. Equivalently, we can say that a
ground state with $n$ fermions is a cycle but not a boundary in $C_n$. This is precisely the definition of
an element of the $n$-th cohomology of $Q$, $H_Q^{(n)}=\ker Q/\im Q$ within $C_n$. Two states $\ket{s_1}$ and $\ket{s_2}$ are
said to be in the same cohomology-class if $\ket{s_1}=\ket{s_2}+ Q \ket{s_3}$ for some state $\ket{s_3}$.
Since a ground state is annihilated by both $Q$ and $Q^{\dag}$, different (i.e. linearly independent) ground
states must be in different cohomology-classes of $Q$\footnote{Let $\ket{s_1}$ and $\ket{s_2}$ be two
linearly independent ground states. It follows that
$Q\ket{s_1}=Q\ket{s_2}=Q^{\dag}\ket{s_1}=Q^{\dag}\ket{s_2}=0$. If we now write $\ket{s_1}=\ket{s_2}+ Q
\ket{s_3}$, we find that $Q^{\dag}\ket{s_1}=Q^{\dag}\ket{s_2}+ Q^{\dag}Q
\ket{s_3}$ and thus $Q^{\dag}Q\ket{s_3}=0$. From this we find $\bra{s_3}Q^{\dag}Q\ket{s_3}=|Q\ket{s_3}|^2=0$
and thus $Q\ket{s_3}=0$. With this we obtain $\ket{s_1}=\ket{s_2}$, which contradicts our assumption that
$\ket{s_1}$ and $\ket{s_2}$ are linearly independent, so we conclude that $\ket{s_1}$ and $\ket{s_2}$ must be
in different cohomology classes.}. Finally, the number of independent ground states is
precisely the dimension of the cohomology of $Q$ and the fermion-number of a ground state is the same as the
grade of the corresponding cohomology-class. Thus the ground states of a supersymmetric theory are in one-to-one
correspondence with the cohomology of $Q$. With the same line of reasoning we may also conclude that the
ground states are in one-to-one correspondence with the homology of $Q^{\dag}$.
Finally, the Euler characteristic, defined in cohomology theory as
\beq\nonumber
\chi \equiv \sum_n \left[ (-1)^n \hbox{dim} H_Q^{(n)}\right],
\eeq
is precisely the Witten index.

\subsection{The 'tic-tac-toe' lemma}\label{tictactoe}
Central to the proof presented in this paper is the 'tic-tac-toe' lemma of \cite{botttu}. Let us decompose
the lattice $S$ into two sublattices $S_1$ and $S_2=S \setminus S_1$ and we write
$Q = Q_1 + Q_2$, where $Q_1$ and $Q_2$ act on $S_1$ and $S_2$ respectively. We can then consider the double
complex $\oplus_n C_n = \oplus_n \oplus_{p+q=n} K_{p,q}$, where $p$ ($q$) is the size of the vertex set on $S_1$
($S_2$). Equivalently, if we define $f_i$ as the number of particles on $S_i$, we have $f_1=p$ and $f_2=q$.
Finally, we have $Q_1: K_{p,q}\rightarrow K_{p+1,q}$ and $Q_2: K_{p,q}\rightarrow K_{p,q+1}$. The
'tic-tac-toe' lemma now tells us that the cohomology of $Q$, $H_Q$, is the same as the cohomology of $Q_1$
acting on the cohomology of $Q_2$, i.e. $H_Q = H_{Q_1}(H_{Q_2}) \equiv H_{12}$, provided that $H_{12}$ has
entries only in one row. That is, $H_{12}$ is non-vanishing only for one value of $q$ (or $f_2$).

\beq\nonumber
\left. \begin{array}{ccccccc}
\vdots & &\vdots & &\vdots & & \\
\uparrow Q_2 & &\uparrow Q_2 & &\uparrow Q_2 & & \\
K_{0,2} & \left. \begin{array}{c} Q_1 \\
\longrightarrow
\end{array} \right. & K_{1,2} & \left. \begin{array}{c} Q_1 \\
\longrightarrow
\end{array} \right. & K_{2,2} & \left. \begin{array}{c} Q_1 \\
\longrightarrow
\end{array} \right. & \cdots\\
\uparrow Q_2 & &\uparrow Q_2 & &\uparrow Q_2 & & \\
K_{0,1} & \left. \begin{array}{c} Q_1 \\
\longrightarrow
\end{array} \right. & K_{1,1} & \left. \begin{array}{c} Q_1 \\
\longrightarrow
\end{array} \right. & K_{2,1} & \left. \begin{array}{c} Q_1 \\
\longrightarrow
\end{array} \right. & \cdots \\
\uparrow Q_2 & &\uparrow Q_2 & &\uparrow Q_2 & & \\
K_{0,0} & \left. \begin{array}{c} Q_1 \\
\longrightarrow
\end{array} \right. & K_{1,0} & \left. \begin{array}{c} Q_1 \\
\longrightarrow
\end{array} \right. & K_{2,0} & \left. \begin{array}{c} Q_1 \\
\longrightarrow
\end{array} \right. & \cdots
\end{array} \right.
\eeq

\section{Statement of main result}\label{mainres}
The main result of this paper can be formulated both in the
physics and mathematics context. We prove the result in the
mathematics context, namely we find the dimensions of the
cohomology for the independence complex on the square lattice
wrapped around a torus. In the physics context this translates to
the statement that we found the total number of ground states for
the supersymmetric model on the square lattice wrapped around a
torus. As we mentioned at the end of section \ref{physintro}, the
solution is found by relating ground states, or equivalently
elements of the cohomology, to tilings of the plane with two types
of rhombi. As was mentioned before, this relation is inspired by
the work of Jonsson \cite{Jonsson1,Jonsson2}. He first introduced
the rhombi when he related the partition sum of hard squares with
activity $z=-1$ to these rhombus tilings. This is precisely the
Witten index for our model on the square lattice, and also the
Euler characteristic of $H_Q$. The Witten index is a lower bound
to the number of ground states. The result we obtain in this paper
gives us, not just this bound, but the total number of ground
states with their respective fermion-number in terms of rhombus
tilings. A rhombus tiling is obtained by tiling the plane with the
rhombi depicted in figure \ref{fig:rhombi}, such that the entire
plane is tiled and the rhombi do not overlap (they can have only a
corner or a side in common). We call the tiles with area 4
diamonds and the ones with area 5 squares.
\begin{figure}[h!]
\begin{center}
\includegraphics[width= .45\textwidth]{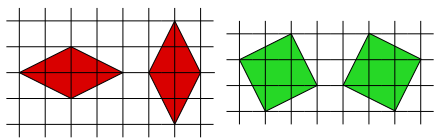}
\caption{The diamonds on the left and squares on the right.}
\label{fig:rhombi}
\end{center}
\end{figure}
We can now state the main result of this paper.

\begin{theorem}\label{tm:HQ}
For the square lattice with periodicities $\vec{v}=(v_1,v_2)$, $v_1+v_2=3p$ with $p$ a positive integer and
$\vec{u}=(m,-m)$, we find for the cohomology $H_Q$
\beq\label{eq:HQ}
N_n=\hbox{dim\ }H_Q^{(n)}= t_n+\Delta_n
\eeq
where $N_n$ is the number of zero energy ground states with $n$ fermions, $t_n$ is the number of rhombus tilings
with $n$ tiles, and
\beq\label{eq:Delta}
\Delta_n = \left\{ \begin{array}{ll}
\Delta \equiv -(-1)^{(\theta_m+1)p}\theta_{d} \theta_{d*} & \textrm{if } n=[2m/3]p\\
0 & \textrm{otherwise,}
\end{array} \right.
\eeq
with $[a]$ the nearest integer to $a$. Finally, $d=\gcd(u_1-u_2,v_1-v_2)$, $d^*=\gcd(u_1+u_2,v_1+v_2)$
and
\beq
\theta_d \equiv \left\{ \begin{array}{ll}
2 & \textrm{if $d=3k$, with $k$ integer}\\
-1 & \textrm{otherwise.}
\end{array} \right.
\eeq
\end{theorem}

As an immediate consequence of this theorem, we obtain for the Euler characteristic (= Witten index =
$Z(\Sigma,-1)$)
\beq\nonumber
\chi \equiv \sum_n \left[ (-1)^n \hbox{dim} H_Q^{(n)}\right]= \sum_n (-1)^n (t_n + \Delta_n).
\eeq
which is precisely the result obtained by Jonsson for the hard squares at activity $z=-1$ \cite{Jonsson1}.

Another direct consequence follows from the area of the tiles. The diamonds have area 4, and thus a tiling
with solely diamonds will contain $L/4$ tiles. This corresponds to an element in the $L/4$-th cohomology and
a ground state with $L/4$ particles. Conversely, a tiling consisting of squares only corresponds to an
element in the $L/5$-th cohomology and a ground state with $L/5$ particles. Continuing this argument for
general tilings with the diamonds and squares, we find on the infinite plane that for all rational numbers
$r \in [\frac{1}{5},\frac{1}{4}]\cap \mathbb{Q}$ the cohomology at grade $rL$ is non vanishing, or,
equivalently, there exists a zero energy ground state with $rL$ particles. This result was obtained
independently by Jonsson \cite{Jonsson2} for the homology of $Q^{\dag}$. In fact, he found that for each
tiling there is a so-called cross-cycle, which is a representative of the homology of $Q^{\dag}$. However, he
could not prove that these cross-cycles are independent, i.e. in different homology classes, nor that they
constitute a basis. A comparison with our result, theorem \ref{tm:HQ}, suggests that the cross-cycles are
indeed independent and span the full homology with the exception of $\Delta_n$ elements at the $n$-th grade.

Finally, the theorem provides insight in the growth behavior of the number of ground states, since this is 
now directly related to the growth behavior of the number of tilings. In \cite{Jonsson3} various results 
on the number of tilings on the doubly periodic square lattice are reported. Here we mention two of these results 
for the case that $\vec{u}=(m,-m)$ and $\vec{v}=(k,k)$.
\begin{itemize}
\item[1] For $m$ and $k$ such that $\gcd(m,k)=1$, there are no rhombus tilings that satisfy the periodicities given by $\vec{u}$ and $\vec{v}$.
\item[2] For $m=3\mu r$ and $k=3 \lambda r$, with $\mu$ and $\lambda$ positive integers and $r$ large, the total number of rhombus tilings $t$ grows as
\beq\nonumber
t \equiv \sum_n t_{n} \sim \frac{9}{2} \frac{4^{\mu r+\lambda r}}{\pi r \sqrt{\mu \lambda}}.
\eeq
\end{itemize}
In the first case it follows that the number of ground states with $n$ particles is given by $\Delta_n$, 
which is non-zero only for $n=[2m/3]p$ given that $2k=3p$. In the second case the number of ground states will show the same 
growth behavior as the number of tilings. This number turns out to be dominated entirely by the number of
tilings with $2 L/9$ tiles. Furthermore, it is noteworthy that the number of tilings grows exponentially with the 
linear dimensions, instead of the area, of the system. It follows that, eventhough the system is highly 
frustrated, this leads only to a sub-extensive ground state entropy. This is in contrast with results 
obtained for the triangular, hexagonal and martini lattices, for which the ground state entropy was found to 
be extensive \cite{FS,vE}.

\section{Proof of main result}\label{proof}

In this section we present the proof of theorem \ref{tm:HQ}. Unfortunately, the proof is quite involved and
consists of several intermediate steps. Here we will give a brief outline of these steps. In section
\ref{tictactoe} we resumed the 'tic-tac-toe' lemma which plays a crucial role in the proof. The lemma relates
$H_Q$ to $H_{Q_1}$ and $H_{Q_2}$ when $Q$ is written as $Q_1+Q_2$. This is achieved by writing the lattice
$S$ as $S_1 \cup S_2$ and letting $Q_i$ act solely on $S_i$. The crucial step is to choose the right
sublattices. It turns out that for the square lattice one should pick a set of disconnected points for $S_1$
and a set of (disconnected) chains for $S_2$ (for the details see definition \ref{def:S1S2}).

First, in section \ref{chain}, we will discuss the cohomology results for a single chain with various boundary conditions. These
results are crucial in the first step of the 'tic-tac-toe' lemma, i.e. computing the cohomology of $Q_2$,
since $Q_2$ acts on a set of chains.

Second, in section \ref{partI}, we consider the square lattice on the plane and on the cylinder, to
illustrate the power of the 'tic-tac-toe' lemma for a relatively simple case. We choose the boundary
conditions in such a way that $H_{Q_2}$ is non-vanishing only for one value of $f_1$ and $f_2$. Consequently,
$H_{12}$ and $H_Q$ are trivially obtained from $H_{Q_2}$.

Finally, we wrap the square lattice around the torus. We then apply the same strategy as in section \ref{partI},
and $H_{Q_2}$ is easily obtained. Unfortunately, however, it has entries in several rows and columns of the
double complex and computing $H_{12}$ is highly non trivial.

As a first step (section \ref{sec:zz}), we compute $H_{12}$ for a
thin torus, such that the $S_2$ sublattice consists of one chain
only. For this case, we then show that $H_Q=H_{12}$, even though
$H_{12}$ has entries in multiple rows. The final step for this
simple case is to relate the elements of $H_Q$ to periodic
sequences of tiles and identify the elements that give rise to
the small number $\Delta$ in (\ref{eq:HQ}).

In the last step (section \ref{sec:pchains}), we finally present the proof of theorem \ref{tm:HQ}. Here the
sublattice $S_2$ consists of an arbitrary number of chains. We proceed as in section \ref{sec:zz} to obtain
$H_{12}$ and each step will be similar, but slightly more involved. Again we find that $H_{12}$ does not have
entries only in one row. In contrast to the thin torus case, however, we find that here $H_Q$ is
contained in but not equal to $H_{12}$. Using the 'tic-tac-toe' procedure, we reduce $H_{12}$ to obtain
$H_Q$. What we find is that all the elements of $H_Q$ can be obtained by concatenating so-called building
blocks and in the final step we map each building block to a sequence of tiles. It follows from this mapping,
that the elements of $H_Q$ map to all possible tilings, again with a small discrepancy $\Delta$, which is
computed in the very last step.

\subsection{The cohomology of $Q$ on the chain}\label{chain}
In the following sections we will often use the cohomology results for the supersymmetric model on the chain.
These results can be found in \cite{FNS,FSdB}, but will be restated here for completeness.

\begin{definition}
An open chain of length $L$ is the graph $G(V,E)$ with vertices $V=\{v_j|j\in {\mbox{\bbold N}},j\leq L\}$ and
edges $E=\{(v_j,v_{j+1})|j \in {\mbox{\bbold N}},j < L\}$. A periodic chain of length $L$ is a cycle defined by
the graph $G(V,E')$ where $E'=E \cup \{(v_L,v_1)\}$.
\end{definition}

In the following $j \in {\mbox{\bbold N}}$ ($j$ can also be zero as long as the number of sites $L$ is
positive).
\begin{theorem}\label{chaincohom}
The cohomology of $Q$ on the periodic chain with $L$ sites has:
\begin{itemize}
\item 2 non-trivial cohomology classes with $j$ fermions if $L=3j$,
\item 1 non-trivial cohomology class with $j$ fermions if $L=3j\pm1$.
\end{itemize}
The cohomology of $Q$ on the open chain with $L$ sites has:
\begin{itemize}
\item 1 non-trivial cohomology class with $j$ fermions if $L=3j$ or $L=3j-1$,
\item 0 non-trivial cohomology classes if $L=3j+1$.
\end{itemize}
\end{theorem}
\begin{proof}
We prove this result for the periodic chain with $L=3j$ sites. We take $S_2$ to
be every third site and the remaining sites $S_1$. Remember that $f_i$ is the number of fermions on
sublattice $S_i$. Consider a single site on $S_2$. If both of the adjacent
$S_1$ sites are empty, $H_{Q_2}$ is trivial: $Q_2$ acting on the empty
site does not vanish, while the filled site is $Q_2$ acting on
the empty site. So the empty site does not belong to the kernel of $Q_2$, whereas the filled site belongs to
the image of $Q_2$. This leads to a vanishing $H_{Q_2}$, unless every
site on $S_2$ is forced to be empty by being adjacent to
an occupied site. There are only two such configurations:
\begin{eqnarray}\label{configs}
\ket{\alpha} &\equiv& \dots \bullet \Box \circ \bullet \Box \circ \bullet \Box \circ \bullet \Box \circ \bullet
\Box \circ \bullet \Box \circ \bullet \Box \circ \dots \nonumber\\
\ket{\gamma} &\equiv& \dots \circ \Box \bullet \circ \Box \bullet \circ \Box \bullet \circ \Box \bullet \circ
\Box \bullet \circ \Box \bullet \circ \Box \bullet \dots
\end{eqnarray}
where the square represents an empty site on $S_2$. Both states $\ket{\alpha}$ and
$\ket{\gamma}$ belong to $H_{12}$: they are closed because $Q_1 \ket{\alpha} =
Q_1 \ket{\gamma} = 0$, and not exact because there are no elements
of $H_{Q_2}$ with $f_1 = f -1$ fermions, where $f = L/3 = j$. By the 'tic-tac-toe' lemma, there
must be precisely two different cohomology classes in $H_{Q}$,
and therefore exactly two ground states with $j$ fermions.
\end{proof}

The proofs for the other cases are completely analogous. In the main proof we will need the representatives of
the non-trivial cohomology classes of $Q$ on the open chain. We will use the following notation: to denote a
configuration with fermions on sites $a$, $b$, $c$, etc. we write $\ket{a,b,c \dots}$.
\begin{lemma}\label{chainrepr}
A representative of the non-trivial cohomology classes of $Q$ on the open chain with $L=3j$ or $L=3j-1$ sites is
\beq
\ket{\phi}\equiv \ket{2,5,8 \dots 3j-1},
\eeq
where on the dots the numbers always increase by three.
\end{lemma}
\begin{proof}
From theorem \ref{chaincohom} we know that the representative has
$j$ fermions. In the case that $L=3j$ it follows that $\ket{\phi}$
is the only configuration with $j$ fermions that belongs to the
kernel of $Q$. Since the dimension of $H_Q$ is one, $\ket{\phi}$
must be a representative of the non-trivial cohomology class.

When $L=3j-1$ there are two configurations that belong to the kernel of $Q$:
$\ket{1,4,7 \dots 3j-2}$ and $\ket{2,5,8 \dots 3j-1}$, however the dimension of $H_Q$ is again just one. It
follows that a linear combination of these two configurations will be in the image of $Q$. In general, two
states $\ket{s_1}$ and $\ket{s_2}$ are in the same cohomology class if one can write
$\ket{s_1}=\ket{s_2}+Q\ket{s_3}$ for some state $\ket{s_3}$. In that case both $\ket{s_1}$ and $\ket{s_2}$
are good representatives of the cohomology class. Since $\ket{\phi}$ itself is not in the image of $Q$ it is
thus a good representative.
\end{proof}

\subsection{The cohomology of $Q$ on the square lattice.
Part I: Tilted rectangles and cylinders}\label{partI}

Let us define $\R(M,N)$ with $M, N \geq 1$ as the subset of $\zs^2$ given by the points $(x,y)$ such
that
\begin{eqnarray}\label{bc}
y\leq x \leq y+M-1 \quad \textrm{and} \quad  -y+1\leq x \leq -y+N.
\end{eqnarray}
This defines a tilted rectangular part of the square lattice.
We can also define $\tilde{\R}(M,N)$ with $M, N \geq 1$ as the subset of $\zs^2$ given by the points
$(x,y)$ such that
\begin{eqnarray}
y\leq x \leq y+M-1 \quad \textrm{and} \quad  -y\leq x \leq -y+N-1.
\end{eqnarray}
Whereas $\tilde{\R}(M,N)$ contains the point $(0,0)$, it is excluded in $\R(M,N)$.
The lattice $\tilde{\R}(M,N)$ can be mapped to a lattice of the former type except when $M$ and $N$ are both
odd. Finally, for $M$ even the cylindrical version $\R_c(M,N)$ can be obtained from $\R(M+1,N)$ by identifying the
vertices $(i,i)$ and $(i+M/2,i-M/2)$.

For the lattices $\R(M,N)$, $\tilde{\R}(M,N)$ and $\R_c(M,N)$ the full cohomology problem has been solved
using Morse theory \cite{BM}. These cases can also be solved using the `tic-tac-toe' lemma. The crucial step
is to choose the right sublattices. We take a set of disconnected sites for $S_1$ and a set of (open or periodic)
chains for $S_2$ (see fig. \ref{fig:bc}).

\begin{definition}\label{def:S1S2}
More formally, for $\R(M,N)$ $S_1$ is the set of points $(x,y)$ that satisfy
\begin{eqnarray}
y\leq x \leq y+M-1 \quad \textrm{and} \quad  -y\leq x \leq -y+N-1 \nonumber\\
\textrm{and} \quad x=-y+3s,
\end{eqnarray}
with $3 \leq 3s \leq N-1$ and $S_2$ is the set of points $(x,y)$ that satisfy
\begin{eqnarray}
y\leq x \leq y+M-1 \quad \textrm{and} \quad  -y\leq x \leq -y+N-1 \nonumber\\
\textrm{and} \quad -y+3p+1 \leq x \leq -y+3p+2,
\end{eqnarray}
with $0 \leq 3p \leq N-3$. The sublattices can be defined similarly for $\tilde{\R}(M,N)$.
\end{definition}

\begin{figure}[h!]
\begin{center}
\includegraphics[width= .6\textwidth]{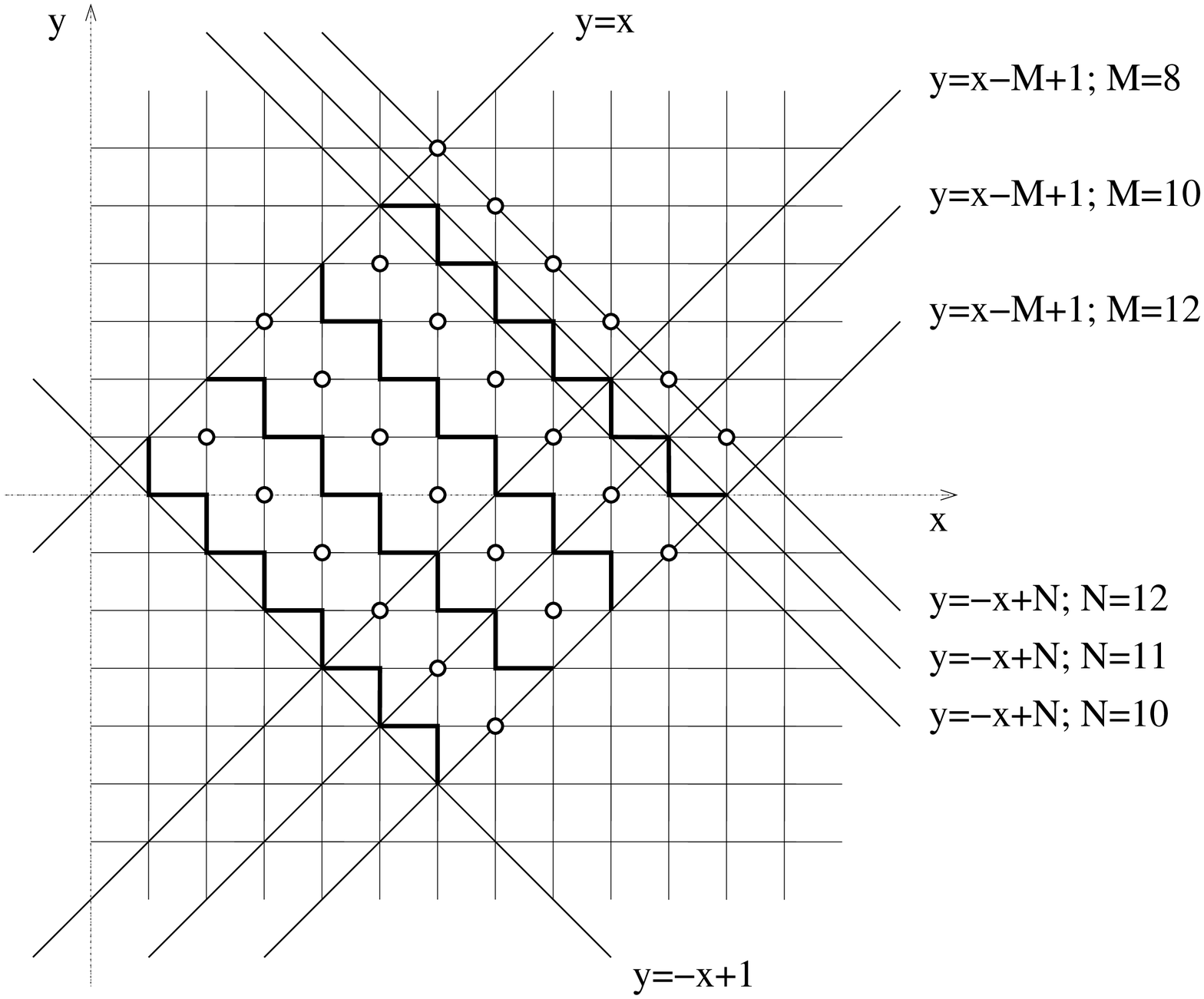}
\caption{Sublattice $S_1$ is indicated by circles and
sublattice $S_2$ is indicated by the fat lines. The bounding lines of $\R(M,N)$ defined in (\ref{bc}) are drawn for various values
of $M$ and $N$. For the cylinder $M$ should be even, but for the rectangle it can be odd as well. One easily
checks that the length of the $S_2$ chain is $M$. Note that for $N=3n+1$ only half of the upper-right $S_2$ chain is
included in $\R(M,N)$.}
\label{fig:bc}
\end{center}
\end{figure}

To solve $H_{Q_2}$ we start from the bottom-left chain.
If a site on $S_1$ directly above this chain is occupied,
we are left with an isolated site on the bottom-left chain (see fig. \ref{fig:singlesite}), leading to a vanishing $H_{Q_2}$ (see section
\ref{chain}). It follows that all sites directly above the bottom-left chain must be empty. Continuing
this argument for subsequent chains one finds that all sites on $S_1$ must be empty. However, in the case
that $N=3l+1$ we have a set of disconnected sites at the top-right that belong to $S_2$. From the previous
argument we obtained that the sites of $S_1$ directly below the top-right sites of $S_2$ have to be empty.
This implies that for $N=3l+1$ $H_{Q_2}$ vanishes.
When $N \neq 3l+1$ we find that all elements in $H_{Q_2}$ have all sites in $S_1$ empty, thus computing
$H_{Q_1}(H_{Q_2})$ is a trivial step. The dimension of $H_Q$ is related to the number of ground
states, or equivalently, the number of non-trivial cohomology classes of $Q$ on the chains that constitute
$S_2$. Note that the length of these chains is $M$ both for the tilted
rectangles as well as for the cylinder. In the first case the chains have open boundary conditions, whereas
in the latter the chains are periodic. Now, the number of non-trivial
cohomology classes of $Q$ for all these cases can be found in theorem \ref{chaincohom}.

\begin{figure}[h!]
\begin{center}
\includegraphics[width= .2\textwidth]{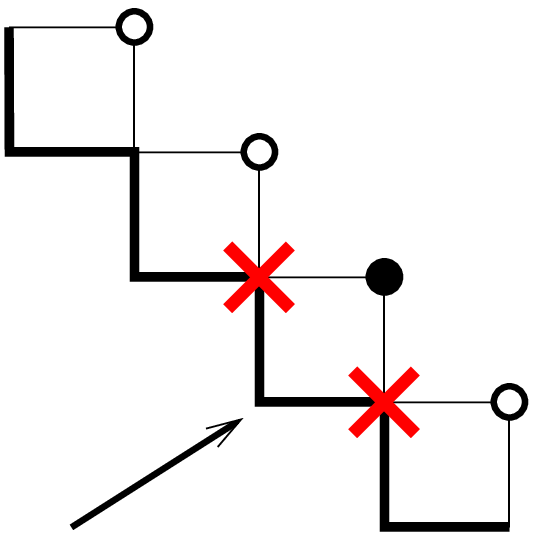}
\caption{A site directly above the bottom-left chain is occupied. This generates an isolated site on the
bottom-left chain.}
\label{fig:singlesite}
\end{center}
\end{figure}

It follows that for the tilted rectangles, $\R(M,N)$ and $\tilde{\R}(M,N)$, with $N \neq 3l+1$ we have
\begin{itemize}
\item[i)] no non-trivial cohomology classes for $M=3p+1$
\item[ii)] one non-trivial cohomology class for $M \neq 3p+1$.
\end{itemize}
For the cylinder, $\R_c(M,N)$, with $N \neq 3l+1$ and $M$ even we have
\begin{itemize}
\item[i)] one non-trivial cohomology class for $M=3p \pm 1$
\item[ii)] $2^K$ non-trivial cohomology classes for $M = 3p$, with $K$ the nearest integer to $N/3$.
\end{itemize}
For $N = 3l+1$ the non-trivial cohomology vanishes for both the rectangle and the cylinder.

\subsection{The cohomology of $Q$ on the square lattice.
Part II: The torus}
We now define the doubly periodic lattices via two linearly independent vectors $\vec{u}=(u_1,u_2)$ and
$\vec{v}=(v_1,v_2)$. We wrap the square lattice around the torus by identifying all points $(i,j)$ with
$(i+k u_1+ l v_1,j+ k u_2+ l v_2)$.
The main result of this paper is that we find the full cohomology of $Q$ on the square lattice with doubly
periodic boundary conditions defined by $\vec{u}=(m,-m)$ and $\vec{v}=(v_1,v_2)$ such that $v_1+v_2=3p$. In
particular, we obtain a direct relation between elements of $H_Q$ and tiling configurations. This relation
allows us to prove theorem \ref{tm:HQ}, that was first conjectured by P. Fendley and is strongly
inspired by the work of Jonsson \cite{Jonsson1,Jonsson2}. It is restated here for convenience.\\

{\em
For the square lattice with periodicities $\vec{v}=(v_1,v_2)$, $v_1+v_2=3p$ with $p$ a positive integer and
$\vec{u}=(m,-m)$, we find for the cohomology $H_Q$
\beq\nonumber
N_n=\hbox{dim\ }H_Q^{(n)}= t_n+\Delta_n
\eeq
where $N_n$ is the number of zero energy ground states with $n$ fermions, $t_n$ is the number of rhombus tilings
with $n$ tiles, and
\beq\nonumber
\Delta_n = \left\{ \begin{array}{ll}
\Delta \equiv -(-1)^{(\theta_m+1)p}\theta_{d} \theta_{d*} & \textrm{if } n=[2m/3]p\\
0 & \textrm{otherwise,}
\end{array} \right.
\eeq
with $[a]$ the nearest integer to $a$. Finally, $d=\gcd(u_1-u_2,v_1-v_2)$, $d^*=\gcd(u_1+u_2,v_1+v_2)$
and
\beq\nonumber
\theta_d \equiv \left\{ \begin{array}{ll}
2 & \textrm{if $d=3k$, with $k$ integer}\\
-1 & \textrm{otherwise.}
\end{array} \right.
\eeq
}\\

Computing the cohomology for these tori is far from trivial. First of all, computing $H_{Q_2}$ does not
imply that all sites on $S_1$ are empty, instead there are many allowed configurations on $S_1$. Secondly,
because of this, computing $H_{Q_1}(H_{Q_2})$ becomes much more involved. Finally, we will see that,
generally, $H_Q$ will be contained in $H_{12}$, but not equal to $H_{12}$.

We will divide the proof into two parts. We start by proving the theorem for a specific torus, defined by
$\vec{u}=(m,-m)$ and $\vec{v}=(1,2)$. This proof will already contain many steps that we use in the proof for
the more general case, however, it will be deprived of certain complications. For instance, here we will find
$H_Q=H_{12}$. As we said, this is not true in general, and in the second part of the proof a substantial part
will be concerned with obtaining $H_Q$ once we have found $H_{12}$.


\subsubsection{A special case: $S_2$ consisting of 1 chain}\label{sec:zz}
In this section we consider the case where $\vec{v}=(1,2)$ and
$\vec{u}=(m,-m)$. It follows that $S_2$ consists of exactly one
periodic chain (see fig. \ref{fig:onechain}).

\begin{figure}[h!]
\begin{center}
\includegraphics[width= .8\textwidth]{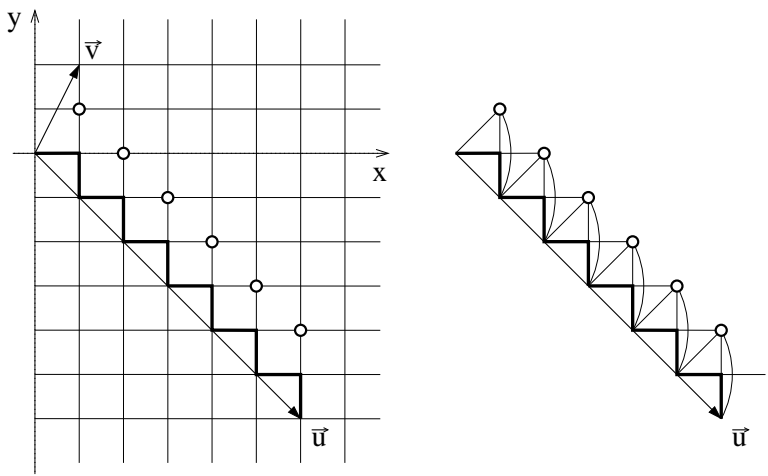}
\caption{The square lattice is wrapped around the torus by imposing periodicities $\vec{v}$ and $\vec{u}$.
Here $\vec{v}=(1,2)$ and $\vec{u}=(m,-m)$, consequently $S_2$ consists of 1 chain. On the right we have
drawn the cylinder, where periodicity in the $\vec{u}$-direction is still implied.}
\label{fig:onechain}
\end{center}
\end{figure}

The proof of theorem \ref{tm:HQ} for this case will consist of 4 steps:
\begin{itemize}
\item[1.] We compute $H_{Q_2}$.
\item[2.] We compute $H_{12}=H_{Q_1}(H_{Q_2})$ and show that its elements can be constructed from a finite
number of building blocks, called motifs. A motif is characterized by a certain configuration on a finite
number of subsequent $S_1$ sites.
\item[3.] We show that $H_{Q}=H_{12}$.
\item[4.] We relate the elements of $H_{Q}$ to tiling configurations by relating each motif to a
small series of tiles.
\end{itemize}


\vspace{0.5cm}
\noindent
{\bf Step 1}

\vspace{0.1cm}
\noindent
First we compute $H_{Q_2}$ and we find the following:
\begin{lemma}\label{lm:notinHQ2}
The cohomology of $Q_2$ consists of all possible configurations on $S_1$ except for configurations with a
multiple of three $S_1$ sites empty between two occupied $S_1$ sites.
\end{lemma}
\begin{proof}
The proof is relatively simple. First note that when a site on
$S_1$ is occupied, it blocks four subsequent sites on the $S_2$
chain (see fig. \ref{fig:onechain}). By occupying sites on $S_1$
the periodic $S_2$ chain is cut into smaller pieces of chain with
open boundary conditions. Consequently, $H_{Q_2}$ vanishes when at
least one of these smaller pieces has length $3p+1$. This happens
when the number of empty sites between two occupied sites on $S_1$
is a multiple of three. We conclude that all configurations on
$S_1$ are allowed except for configurations with a multiple of
three sites empty between two occupied sites. 
\end{proof}

\vspace{0.5cm}
\noindent
{\bf Step 2}

\vspace{0.1cm}
\noindent
This step in the proof is the most involved. In the next section, where we prove theorem \ref{tm:HQ} in all
generality, we will often refer back to the results obtained in this step. In this step we compute
$H_{Q_1}(H_{Q_2})$, where $H_{Q_2}$ was obtained in the previous step. Let us define $f_1$ and $f_2$ as the number of fermions on $S_1$ and $S_2$
respectively. Furthermore we shall adopt the following notation: an empty site on $S_1$ is denoted by 0 and
an occupied site is denoted by 1. A configuration on $S_1$ can then be written as a series of 1's and
0's. In the following we shall consider all possible types of configurations on $S_1$ that belong to
$H_{Q_2}$ and we shall investigate if they also belong to $H_{12}$.

If we consider a configuration on $S_1$, we note that there have to be at least two adjacent, empty $S_1$
sites to allow for $f_2>0$. This is because two adjacent empty sites leave an open chain of two sites unblocked on $S_2$
and this has an element in $H_{Q_2}$ with $f_2=1$. A typical configuration will thus consist of alternating segments, where a
segment is a sequence of $S_1$ sites. The segments are characterized
by the number of fermions on the part of the $S_2$ chain corresponding to the segment, this is either zero or
greater than zero. In a segment with $f_2>0$ all $S_1$ sites
are empty and it contains at least two sites. On the other hand, a segment with $f_2=0$ can have empty sites
on $S_1$, but the empty sites cannot be adjacent. Finally, a segment with $f_2=0$ will always start and end
with an occupied $S_1$ site. We will call this pair of occupied sites a pair of bounding sites. Note that a
segment with $f_2=0$ can consist of a single occupied site, in that case the bounding sites fall on top of
each other and the pair of bounding sites is just this one site.

\begin{example}
Consider the configuration "$1101101010000100$", there are two segments with $f_2=0$, formed by the first
nine sites and the fourteenth site respectively. There are also two segments with $f_2>0$ constituted by the rest of the
sites. Finally, the first and the ninth site form a pair of bounding sites.
\end{example}

First, we consider the segments with $f_2>0$.
\begin{lemma}\label{Q1zerowithinHQ2}
$Q_1$ acting on a segment with $f_2>0$ gives zero within $H_{Q_2}$.
\end{lemma}
\begin{proof}
Suppose this segment between a pair of bounding sites consists of $l$ empty $S_1$ sites. The corresponding
$S_2$ chain then has length $L=2l-2$. Since $l=3k\pm1$, we find $L=6k$ or $L=6k-4$. For these chain lengths
the elements of the cohomology of $Q_2$ have $2k$ and $2k-1$ fermions respectively. We now distinguish two
cases: a) $Q_1$ acts on a site at the boundary of the segment, b) $Q_1$ acts on a site away from the
boundary.

a) In this case the length of the $S_2$ chain in the new configuration is $L'=L-2$. Thus $L'=6k-2$ or
$L'=6k-6$. On the new chain there are still $2k$ or $2k-1$ fermions respectively. However, theorem
\ref{chaincohom} states that the cohomology for chain length $6k-2$ ($6k-6$) vanishes at all fermion numbers
except $f=2k-1$ ($f=2k-2$). Thus the new configuration does not belong to $H_{Q_2}$ and it follows that this
action of $Q_1$ within $H_{Q_2}$ gives zero.

b) In this case the action of $Q_1$ cuts the $S_2$ chain into two smaller chains of lengths $L_1'$ and
$L_2'$. Their total length is $L_1'+L_2'=L-4$, since the occupied $S_1$ site now blocks 4 sites on the
$S_2$ chain. For $L=6k$ we have $L_1'=3k_1$ and $L_2'=3k_2+2$ or $L_1'=3k_1+1$ and
$L_2'=3k_2+1$, where in both cases $k_1+k_2=2k-2$. For the latter case $H_{Q_2}$ vanishes at all grades. In
the first case $H_{Q_2}$ is non-vanishing only for $f=k_1+k_2+1=2k-1$. However, the number of fermions on the
$S_2$ chains in the new configuration is $f=2k$ and thus is does not belong to $H_{Q_2}$. Similarly, one
finds that for $L=6k-4$, the new configuration does not belong to $H_{Q_2}$. Again we obtain that this
action of $Q_1$ within $H_{Q_2}$ gives zero.

Finally, if the segment with $f_2>0$ extends over the entire system, we are always in the case considered
under b). However, the original chain length can now also be $L=6k-2$ with $2k-1$ fermions on it. Under the
action of $Q_1$ we obtain a new chain of length $L'=6k-6$, which has a non-vanishing cohomology if and only
if $f=2k-2$. So also in this case we find that the action of $Q_1$ gives zero within $H_{Q_2}$.

\end{proof}

Second, we consider the segments with $f_2=0$.
\begin{lemma}\label{lm:HQ1}
$H_{Q_1}(H_{Q_2})$ vanishes when the number of $S_1$ sites between any pair of bounding sites in a segment
with $f_2=0$, is $3p+1$ and it contains one element otherwise.
\end{lemma}

\begin{example}
Consider a configuration with one empty site between a pair of bounding sites: "$101$", this is not an
element of $H_{Q_1}(H_{Q_2})$, since $Q_1$ on this configuration gives "$111$", which is also in $H_{Q_2}$.
Now consider two sites between a pair of bounding sites. Then there are two configuration with one empty
site: "$1011$" and "$1101$" and one configuration with all sites occupied "$1111$" (remember that the
configuration "$1001$" does not have $f_2=0$). It follows that $Q_1$ acting on ("$1101$" $-$
"$1011$") gives $2$"$1111$", whereas $Q_1$ acting on ("$1101$" $+$ "$1011$") gives zero\footnote{Note
that the fermionic character of the particles is reflected in the sign here: $Q_1$ acting on "$1011$"
gives $-$"$1111$", whereas $Q_1$ acting on "$1101$" gives $+$"$1111$". In the first case the particle
is created on position 2 and has to hop over the particle at position 1, this gives a minus sign, in the
second case the new particle is created at position 3 and thus has to hop over two particles, giving rise
to no overall sign change. Also note that the states are not properly normalized, but this is not important
for the argument.}. Consequently, we find that $H_{Q_1}(H_{Q_2})$ consists of one element: the sum of
the configurations with $f_1=3$.
\end{example}

\begin{proof}
We can solve $H_{Q_1}(H_{Q_2})$ for an arbitrary number of sites
between a pair of bounding sites, by realizing that this problem
can be mapped to the normal chain. For the normal chain no two
fermions can be adjacent, whereas here no two empty $S_1$ sites
can be adjacent. So we can map empty $S_1$ sites to fermions on
the chain and occupied $S_1$ sites to empty sites in the normal
chain. Finally, $Q_1$ is mapped to $Q^{\dag}$ on the normal chain.
For the chain $H_{Q^{\dag}}$ (which has the same dimension as
$H_{Q}$) vanishes when the length of the chain is $3p+1$ and it
contains one element otherwise. So here we have that
$H_{Q_1}(H_{Q_2})$ vanishes when the number of sites between two
occupied sites is $3p+1$ and it contains one element otherwise.

\end{proof}
For a segment with $f_2=0$, let us denote the representative of $H_{Q_1}(H_{Q_2})$ by the pair of
bounding sites with dots in between, for example we denote ("$1101$" $+$ "$1011$") by "$1\cdot \cdot 1$".
Even though, this is now a sum of configurations, we will still refer to this simply as a
configuration. It follows that, for a segment with $f_2=0$, two types of configurations are allowed.
The two types can be distinguished by containing $3s-1$ dots or $3s$ dots. Examples of
the first type are: "$1$", "$1 \cdot \cdot 1$", "$1 \cdot \cdot \cdot \cdot \cdot1$", etc. Note that the
configuration with $s=0$, and thus with -1 dots between the pair of bounding sites, is "$1$". Examples of
the second type are: "$11$", "$1 \cdot \cdot \cdot 1$", "$1 \cdot \cdot \cdot \cdot \cdot \cdot 1$", etc.

Combining lemma's \ref{Q1zerowithinHQ2} and \ref{lm:HQ1}, we find
that $H_{Q_1}(H_{Q_2})$ is spanned by all
configurations that can be formed by concatenating the following motifs:\\
"$000$"\\
"$1\cdot_{3s-1}100$"\\
"$1\cdot_{3s}100$"\\
"$1\cdot_{3s-1}10000$"\\
"$1\cdot_{3s}10000$"\\
where $\cdot_{3s}$ means $3s$ dots and, as before, "$1\cdot_{3s-1}1$" with $s=0$ means "$1$".

Finally one can also have all zeroes for any length and all dots for any length. Note that if the number of
$S_1$ sites is a multiple of three, that both the configuration with all zeroes and the one with all dots
account for two linearly independent elements of $H_{12}$. This is because the cohomology of $Q$ acting the
periodic chain with length a multiple of three has dimension two (see theorem \ref{chaincohom}).

\begin{example}\label{ex:zz10}
As an example, suppose we have $\vec{v}=(1,2)$ as always and $\vec{u}=(10,-10)$. This implies that $S_1$
consists of 10 sites and with the defined motifs it follows that the following 12 elements belong to
$H_{12}$: "$1100 000 000$", "$1100 000 100$", "$1100 100 000$", "$1100 100 100$", "$1100 110000$",
"$1100\ 1 \cdot \ \cdot 100$", "$1 \cdot \cdot \ \cdot \cdot \cdot \ \cdot 100$", "$1 \cdot \cdot \ \cdot
100 \ 000$", "$1 \cdot \cdot \ \cdot 100 \ 100$", "$10000 10000$", "$0000000000$" and "$\cdot \cdot \cdot
\cdot \cdot \cdot \cdot \cdot \cdot \cdot$". Note that the first nine motifs have periodicity 10 and thus
account for ten elements of $H_{12}$ each, whereas the motif "$10000 10000$" has periodicity 5 and the
last two motifs have periodicity 1. For each element one can easily compute the number of fermions and it
follows that the first nine motifs have 7 fermions, the motif "$10000 10000$" has 6 fermions and the
last two motifs have again 7 fermions. So in total we have 92 elements in $H_{12}$ with 7 fermions and 5
elements with 6 fermions.
\end{example}

\vspace{0.5cm}
\noindent
{\bf Step 3}

\vspace{0.1cm}
\noindent
In the previous step we have obtained $H_{12}$ for
$\vec{v}=(1,2)$. In this step we show that in this case this is
equal to the cohomology of $Q$. We do this via the 'tic-tac-toe'
procedure \cite{botttu}. That is, we act on a configuration, say
$\ket{\psi}$, with $Q$. The $Q_2$ part will automatically give
zero, but the $Q_1$ part not necessarily, since we no longer
restrict ourselves to the subspace $H_{Q_2}$. If it does give
zero, we know that the configuration belongs to the kernel of $Q$.
The configuration will thus belong to $H_Q$ unless it also belongs
to the image of $Q$. In that case, another configuration will map
to this configuration at the end of the 'tic-tac-toe' procedure.
So we continue with the configurations, $\ket{\psi_0}$, that do
not belong to the kernel of $Q_1$. Since the image of
$\ket{\psi_0}$ does not belong to $H_{Q_2}$ and it does belong to
the kernel of $Q_2$, it must also belong to the image of $Q_2$. So
we can write $Q \ket{\psi_0}=Q_2 \ket{\phi}$, for some
configuration $\ket{\phi}$. Now let us define a new state
$\ket{\psi_1} \equiv \ket{\psi_0}- \ket{\phi}$. It then follows
that $Q \ket{\psi_1}=- Q_1 \ket{\phi}$. If this is zero, we have
found that the state $\ket{\psi_1}$ belongs to the kernel of $Q$.
If it is non-zero we proceed as before: we try to find a
configuration $\ket{\chi}$, such that $Q_1 \ket{\phi}=Q_2
\ket{\chi}$ and define a new state $\ket{\psi_2}\equiv
\ket{\psi_0}- \ket{\phi}+\ket{\chi}$, etc. This procedure ends,
either when we have found a state $\ket{\psi_n}$ such that $Q
\ket{\psi_n}=0$, or when $Q \ket{\psi_n}=\ket{\tilde{\psi}}$ with
$\ket{\tilde{\psi}}$ an element of $H_{Q_1}(H_{Q_2})$. In the
latter case, we say $\ket{\psi_0}$ maps to $\ket{\tilde{\psi}}$ at
the end of the 'tic-tac-toe' procedure and we conclude that
neither $\ket{\psi_0}$ nor $\ket{\tilde{\psi}}$ belong to $H_Q$.

For the case we consider in this section, we will show that for
each element $\ket{\psi_0}$, there is an element $\ket{\psi_n}$
that belongs to the kernel of $Q$. So for each element in $H_{12}$
we can find a corresponding element in $H_Q$, thus we obtain
$H_Q=H_{12}$. In the next section, however, we will see that this
is not true for general boundary conditions. We will then find
that after several steps in the 'tic-tac-toe' procedure we map
certain configurations in $H_{12}$ to other configurations in
$H_{12}$. It follows that the first do not belong to the kernel of
$Q$ and the latter belong to the image of $Q$. In that case $H_Q$
is strictly smaller than $H_{12}$.

\begin{lemma}\label{lem:allinkerQ1}
$H_Q=H_{12}$ for $\vec{v}=(1,2)$ and $\vec{u}=(m,-m)$.
\end{lemma}
\begin{proof}
For the segments with $f_2=0$, we found that $Q_1$ vanishes if we
choose the states represented by the dots such that they are
ground states of the normal chain with empty and occupied sites
exchanged (see lemma \ref{lm:HQ1}). For the segments with $f_2>0$,
we know from lemma \ref{Q1zerowithinHQ2} that the new
configuration always belongs to the image of $Q_2$. That is, $Q_1
\ket{\psi_0}=Q_2 \ket{\phi}$, for some configuration $\ket{\phi}$
if $Q_1$ acts on a segment with $f_2>0$. So we can define a new
configuration $\ket{\psi_1} \equiv \ket{\psi_0}- \ket{\phi}$, such
that $Q \ket{\psi_1}=- Q_1 \ket{\phi}$. Now $Q_1$ either acts on a
different segment with $f_2>0$, in which case the new
configuration again belongs to the image of $Q_2$, or it acts on
the same segment. In the latter case the new configuration is
cancelled by the same configuration in which the two $S_1$ sites
are occupied in the reverse order due to the fermionic character
of the particles. It thus follows that the 'tic-tac-toe' procedure
always gives zero after as many steps as there are segments with
$f_2>0$. 
\end{proof}

\vspace{0.5cm}
\noindent
{\bf Step 4}
\nopagebreak

\vspace{0.1cm}
\noindent
In this final step we will show that the dimension of $H_{Q}$ (and
the fermion number of each state) can be computed by counting all
tiling configurations (and the number of tiles per configuration)
with the four types of tiles depicted in figure \ref{fig:rhombi}.
For the boundary conditions we consider here, the tilings reduce
to a single layer sequence of only two types of tiles. Namely the
two tiles that respect the boundary condition $\vec{v}=(1,2)$.
These tiles have two edges parallel to $(1,2)$ and then the
diamond has the other two edges parallel to $(1,-2)$ whereas the
square has the other edges parallel to $(2,-1)$ (see fig.
\ref{fig:zztiles}). Given the sublattices $S_1$ and $S_2$ there
are three types of vertices: the ones that belong to $S_1$, the
lower left sites of the $S_2$-chain and the upper right sites of
the $S_2$ chain. It follows that the diamond has one of three
types of edges along the $(1,-2)$ direction and a matching type of
edge along the $(1,2)$ direction, the square can have one of three
different types of edges along the $(2,-1)$ direction and a
matching type of edge along the $(1,2)$ direction. We conclude
that we have 6 types of tiles, depicted in figure
\ref{fig:tiletypeszigzag}.

\begin{figure}[h!]
     \centering
     \subfigure[]
     {\label{fig:zztiles}\includegraphics[width=.18\textwidth]{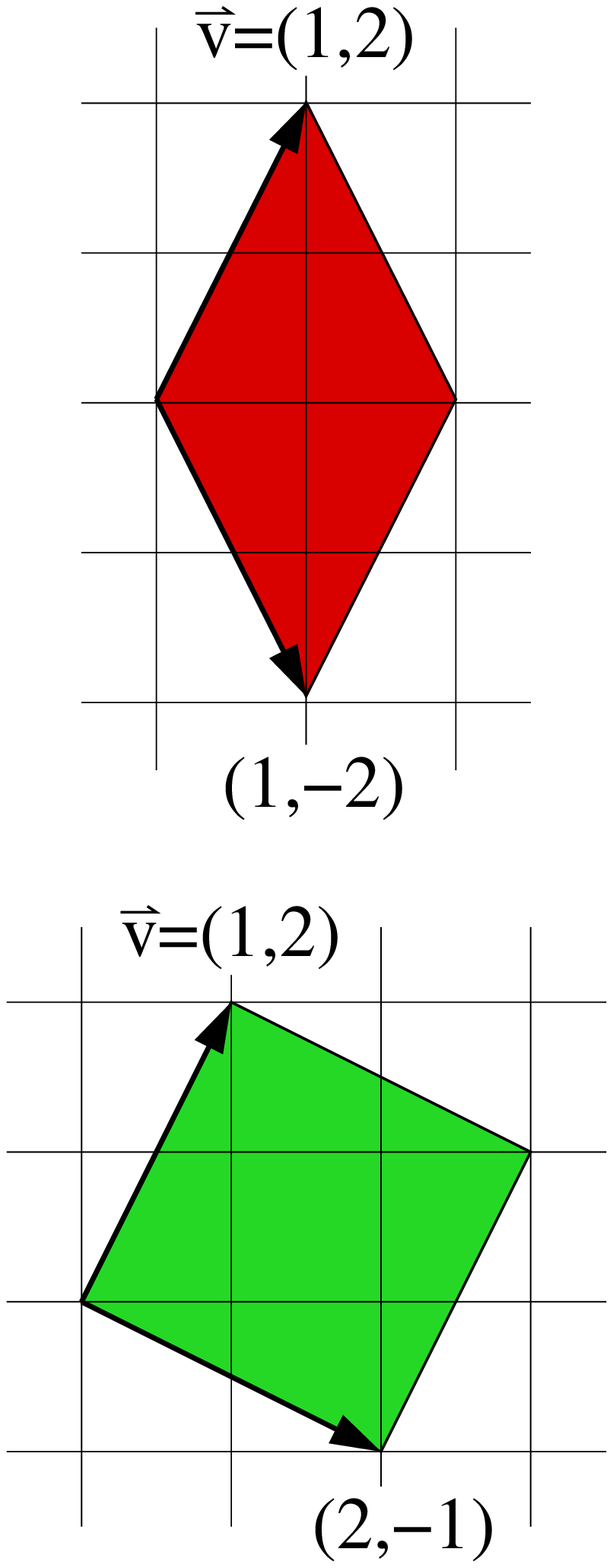}}
     \subfigure[]
     {\label{fig:tiletypeszigzag}\includegraphics[width=.39\textwidth]{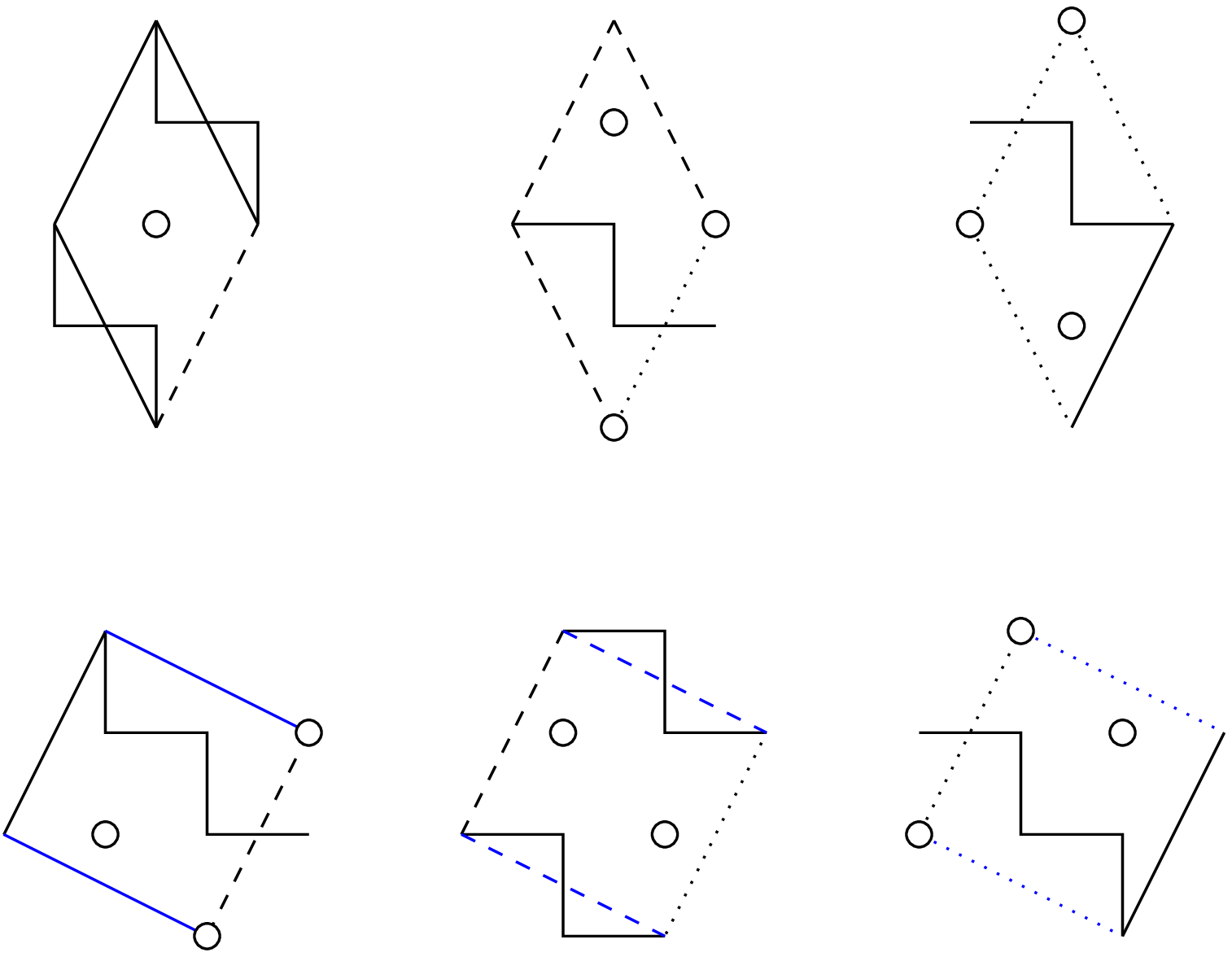}}
     \caption{Types of tiles we use to tile the square lattice with periodicities $\vec{v}=(1,2)$ and
     $\vec{u}=(m,-m)$. \ref{fig:zztiles} Shows the diamond and square that respect $\vec{v}=(1,2)$.
     \ref{fig:tiletypeszigzag} Shows the three types of diamonds and squares given the sublattices $S_1$ and $S_2$.}
     \label{fig:zztilesall}
\end{figure}

To establish theorem \ref{tm:HQ} we map each of the motifs
obtained in step 2 to a unique sequence of tiles. The mapping for
the four basis motifs, "$100$", "$1100$", "$10000$" and
"$110000$", is shown in figure \ref{fig:tilingmapzigzag}. Remember
that each motif is modulo the addition of 3 zeroes and modulo the
insertion of 3 dots. In terms of tilings, we find that each basis
motif can be followed by an arbitrary repetition of the tiling
corresponding to the 3 zeroes (see fig. \ref{fig:zeroesdotszz}).
On the other hand, insertions of multiples of 3 dots correspond to
inserting multiples of the tiling shown in figure
\ref{fig:zeroesdotszz} at the dotted line along the $(1,2)$
direction in the basis motifs. Some examples are shown in figure
\ref{fig:dotsinserted}. Note that here we cannot easily write the
motif of dots directly in the tiles (see fig.
\ref{fig:dotsinserted}), however, the mapping is still
unambiguous.

\begin{figure}[h!]
\begin{center}
\includegraphics[width= .9\textwidth]{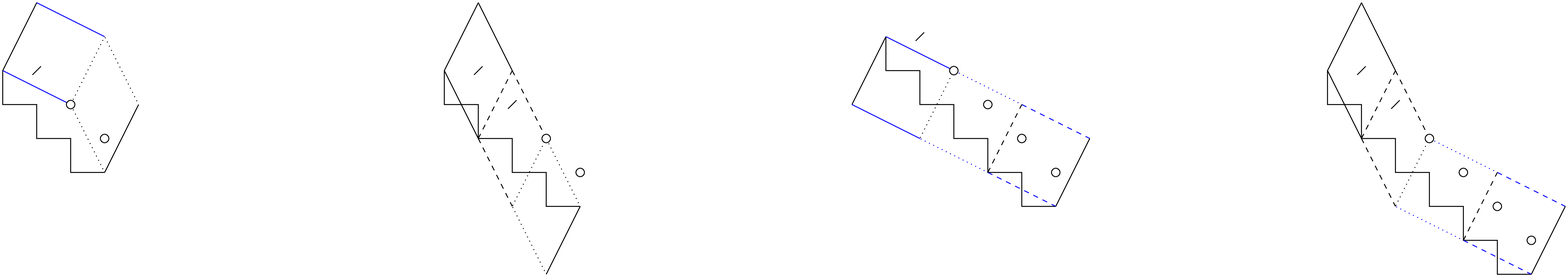}
\caption{The four basis motifs and the corresponding sequences of tiles.}\label{fig:tilingmapzigzag}
\end{center}
\end{figure}

\begin{figure}[h!]
     \centering
     \subfigure[]%
     {\label{fig:zeroesdotszz}\includegraphics[height=4cm]{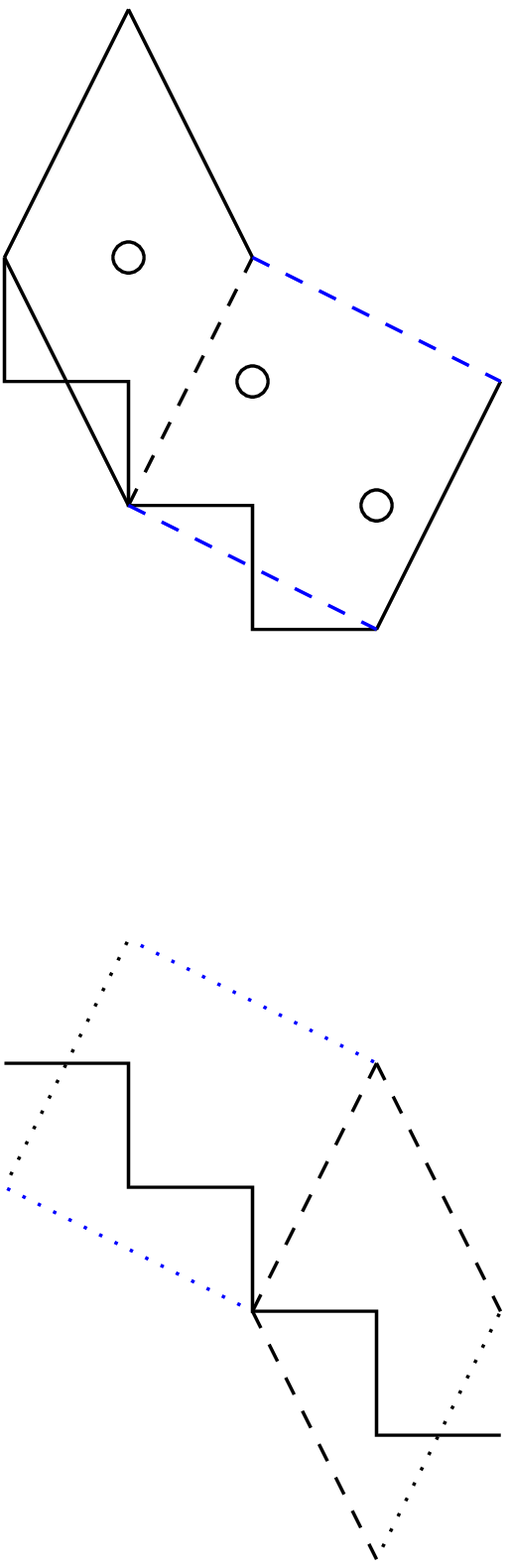}}
     \hspace{2cm}
     \subfigure[]%
     {\label{fig:dotsinserted}\includegraphics[height=5cm]{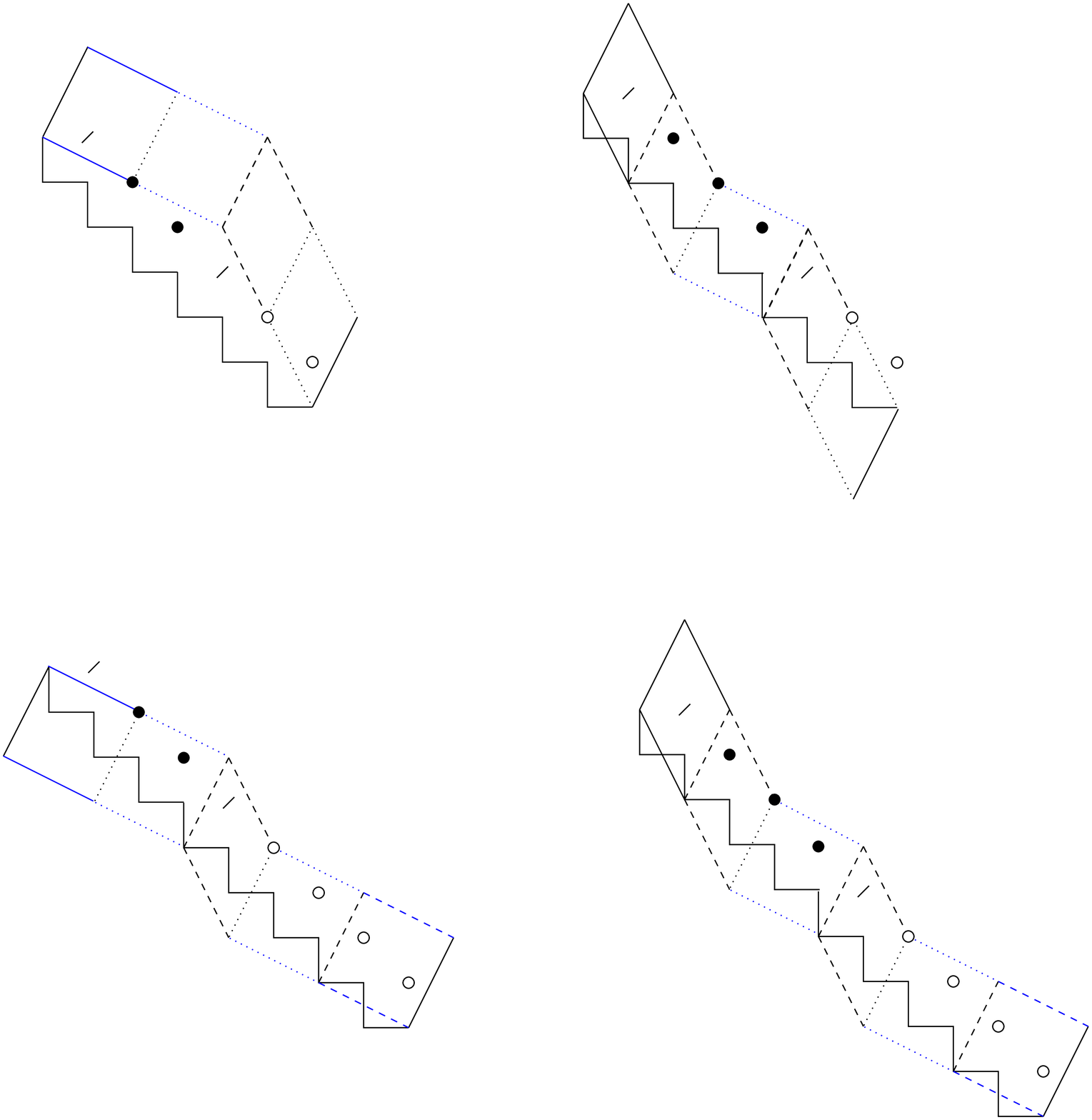}}
     \caption{On the left we show the sequences of tiles that correspond to the motifs with 3 zeroes (top)
     and 3 dots (bottom). The addition of 3 zeroes to a basis motif corresponds to attaching the
     sequence of tiles corresponding to the 3 zeroes to the sequence of tiles corresponding to the basis
     motif. An insertion of 3 dots in a basis motif corresponds to inserting the sequence of tiles corresponding to the 3
     dots at the dotted line along the $(1,2)$ direction in the sequence of tiles corresponding to the basis
     motif. The insertions of 3 dots in each of the four basis motifs and the corresponding tiling are shown on
     the right as examples. From these examples it is clear that we cannot write the 3 dots directly in the
     corresponding sequence of tiles. However, the mapping is still unambiguous.}
\end{figure}

Let us determine the number of fermions per motif. First of all, in a segment with $f_2>0$, the number of
fermions is determined by the length of the corresponding $S_2$ chain. It is easily verified, that for a
segment with $n$ empty $S_1$ sites the corresponding chain has length $2n-2$. Moreover, from theorem
\ref{chaincohom}, we know that an element in the cohomology of $Q$ on a chain with length $L=2n-2$ contains
$[(2n-2)/3]$ fermions, where $[a]$ is the nearest integer to $a$. Similarly, we find that a segment with $k$
dots contains $[2k/3]$ fermions. Thus a segment with $f_2=0$, consisting of $k$ dots and a the pair of bounding
sites, contains $[2k/3]+2$ fermions. From these formulae we find for the four basis motifs "$100$", "$1100$",
"$10000$" and "$110000$", that they contain 2, 3, 3 and 4 fermions respectively. Furthermore, an insertion of
3 zeroes, corresponds to increasing $n$ by 3, and thus increasing the number of fermions, $[(2n-2)/3]$, by 2.
Equivalently, inserting 3 dots corresponds increasing $k$ by 3, and thus again increasing the number of fermions,
$[2k/3]$, by 2. If we compare this to the number of tiles in the tilings that correspond to these motifs, we
find that they exactly agree. Furthermore, the number of sites in a motif is given by three times the number
of $S_1$ sites in a motif, since there are 2 $S_2$ sites for every $S_1$ site. On the other hand, for the tiles
we find that the area of the diamond is 4 and the area of the square is 5. It is now easily verified that the
number of fermions per site for the motifs is the same as the number of tiles per area for the corresponding
tiling. Thus we find that, not only is the number of elements in the cohomology of $Q$
directly related to the number of tilings with the two tiles of figure \ref{fig:zztiles}, but also the number
of fermions for each element corresponds to the number of tiles in the corresponding tiling.

One can verify that with these sequences of tiles,
and the rules for concatenating them, one can obtain every possible tiling. Each tile can be preceded by a
certain type of square and diamond and it can be followed by another type of square and diamond. In total
this gives four possibilities for the surrounding neighbors. It can be checked that for each tile all four
possibilities can be constructed with the given sequences of tiles and the rules for concatenating them.

Finally, the configurations with all zeroes or all dots account for the extra term in (\ref{eq:HQ}) in
theorem \ref{tm:HQ} repeated here for convenience:
\beq
\Delta_i \equiv \left\{ \begin{array}{ll}
-(-1)^{(\theta_m+1)p}\theta_{d} \theta_{d*} & \textrm{if } i=[2m/3]p\\
0 & \textrm{otherwise.}
\end{array} \right.
\eeq
Remember that
\beq
\theta_d \equiv \left\{ \begin{array}{ll}
2 & \textrm{if $d=3k$, with $k$ integer}\\
-1 & \textrm{otherwise}
\end{array} \right.
\eeq
and with $\vec{v}=(1,2)$ and $\vec{u}=(m,-m)$ we have $p=1$, $d=\gcd(u_1-u_2,v_1-v_2)=\gcd(2m,-1)=1$ and
$d^*=\gcd(u_1+u_2,v_1+v_2)=\gcd(0,3)=3$. It follows that the extra term is $-2$ for $m=3n$ and $+2$
otherwise.

Let us consider the configuration with all zeroes, which clearly has periodicity 1. If the number of zeroes is a
multiple of three, i.e. $m=3n$, the configuration accounts for 2 ground states, otherwise it accounts for
1 ground state. The number of fermions in this configuration is $i=[2m/3]$, i.e. the nearest integer to
$2m/3$. From the mapping (fig. \ref{fig:zeroesdotszz}) it is clear that the configuration corresponds to a tiling
with periodicity 3 if $m=3p$. If $m\neq3p$, however, there is no corresponding
tiling. Exactly the same holds for the configuration with all dots. It follows that for $m=3p$ the tilings
overcount the number of ground states by 2 and for $m\neq3p$ the tilings fail to count 2 ground states.

Note that the choice of sublattices $S_1$ and $S_2$ has increased the number of tilings unrelated by a
lattice symmetry by a factor of three
(see fig. \ref{fig:zztilesall}). Indeed when computing the number of ground states with the motifs given
in step 2 it turns out that one discovers each tiling three times (given that the tiling is not completely
uniform, that is all diamonds or all squares).

\begin{example}
Let us consider the case of example \ref{ex:zz10} again. So we have $\vec{v}=(1,2)$ and $\vec{u}=(10,-10)$.
One possibility is to cover the lattice with 6 squares. This tiling has a unit cell of size 5 and thus
this tiling accounts for 5 ground states. The number of tiles is 6 and thus the ground states will have 6
fermions. We can also cover the lattice with 2 squares and 5 diamonds. The 2 squares can be placed between
the diamonds in three independent ways. Each of these three tilings has a unit cell of size 30 and consists
of 7 tiles, so they account for 90 ground states with 7 fermions.

We compare this with the 12 configurations found in example
\ref{ex:zz10}. The motif "$10000 10000$" has periodicity 5 and
accommodates 6 fermions, so this corresponds to the uniform tiling
with all squares. The configurations with all zeroes and all dots
account for two ground states with 7 fermions and have no
corresponding tiling. Finally, there are 9 configurations with
periodicity 10 and 7 fermions, which account for 90 ground states.
Using the mapping given in figure \ref{fig:tilingmapzigzag}, we
find that these configurations can be split into three groups of
three, each group corresponding to one of the tilings with 2
squares and 5 diamonds. For example the motifs "$1100 000 100$",
"$1100 110000$" and "$1100\ 1 \cdot \ \cdot 100$" correspond to
the tiling where the two squares are adjacent. They can be
distinguished by considering for example the first of the two
squares. In each motif it will be of a different type, where the
three types are given in figure \ref{fig:tiletypeszigzag} (see
fig. \ref{fig:zz10}).
\end{example}

\begin{figure}[h!]
\begin{center}
\includegraphics[width= .65\textwidth]{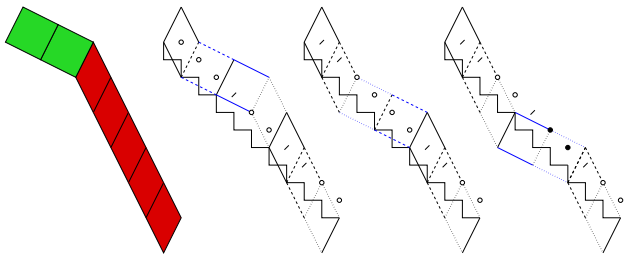}
\caption{The square lattice with periodicities $\vec{v}=(1,2)$ and $\vec{u}=(10,-10)$ can be tiled with 2 squares
and 5 diamonds. One of these tilings, with the two squares adjacent is shown on the left. The choice of
sublattices splits this tiling into three tilings. These three tilings and their corresponding motifs are shown on the right.}
\label{fig:zz10}
\end{center}
\end{figure}


\subsubsection{The general case: $S_2$ consisting of $p$ chains}\label{sec:pchains}
In the previous section we had $\vec{v}=(1,2)$. In this section we relax this condition to
$\vec{v}=(v_1,v_2)$ with $v_1+v_2=3p$ with $p$ a positive integer. It follows that we get $p$ $S_2$
chains with their accompanying $S_1$ sites stacked on top of each other. For this situation we will prove
theorem \ref{tm:HQ}. The proof consists of 5 steps:
\begin{itemize}
\item[1.] We compute $H_{Q_2}$.
\item[2.] We compute $H_{12}=H_{Q_1}(H_{Q_2})$.
\item[3.] We compute $H_{Q}$ starting from $H_{12}$ via the 'tic-tac-toe' procedure.
\item[4.] We relate the elements of $H_{Q}$ to tiling configurations by relating each motif to a
small series of tiles.
\item[5.] We compute $\Delta_i$.
\end{itemize}

\vspace{0.5cm}
\noindent
{\bf Step 1}

\vspace{0.1cm}
\noindent
As in the previous section we shall start by computing the cohomology of $Q_2$. We will define two types of
configurations that do not belong to $H_{Q_2}$ and then find that $H_{Q_2}$ consists of all configurations
except these two types.

\begin{lemma}\label{lm:notinHQ2-a}
A configuration that contains an occupied site $(k,l)$ on the $S_1$ lattice, such that the sites $(k+1,l+2)$
and $(k+2,l+1)$ and/or the sites $(k-1,l-2)$ and $(k-2,l-1)$ are empty, does not belong to $H_{Q_2}$.
\end{lemma}
\begin{proof}
It is easily verified (see fig. \ref{fig:lmnotinHQ2a}) that in
this configuration the $S_2$ sublattice contains the isolated
site(s) $(k+1,l+1)$ and/or $(k-1,l-1)$. This site can be either
occupied or empty, which leads to a vanishing $H_{Q_2}$. 
\end{proof}

\begin{figure}[h!]
\begin{center}
\includegraphics[width= .2\textwidth]{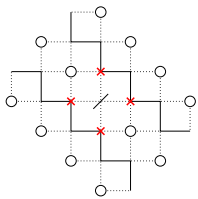}
\caption{A configuration is shown where an occupied $S_1$ site is surrounded by empty sites. This sites
isolates a site on the $S_2$ chains directly below and above the site.}
\label{fig:lmnotinHQ2a}
\end{center}
\end{figure}

Note that in the previous section this situation never occurred because for each occupied site $(k,l)$, the
sites $(k+1,l+2)$ and $(k-1,l-2)$ were automatically occupied due to the boundary condition set by
$\vec{v}=(1,2)$. The second type of configuration that does not belong to $H_{Q_2}$ follows from a
generalization of lemma \ref{lm:notinHQ2}. Remember that occupying $S_1$ sites causes the $S_2$ chains to
break into smaller open chains. The length of these open chains now depends on the number of empty $S_1$
sites directly below and above the $S_2$ chain. For an example see figure \ref{fig:notinHQ2}.

\begin{lemma}\label{lm:notinHQ2-b}
If, for a certain configuration, the sum of the number of empty $S_1$ sites directly below and above an open $S_2$
chain is $3s$, the configuration does not belong to $H_{Q_2}$.
\end{lemma}
\begin{proof}
It is easily verified that the open $S_2$ chain corresponding to the $3s$ empty $S_1$ sites has length
$3(s-1)+1$. This leads to a vanishing $H_{Q_2}$. 
\end{proof}

\begin{figure}[h!]
\begin{center}
\includegraphics[width= .4\textwidth]{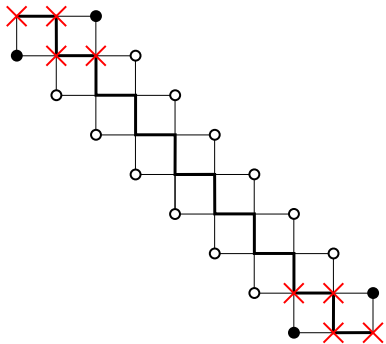}
\caption{Part of a configuration is shown. The number of empty $S_1$ sites directly below and above the $S_2$
chain is 12. The $S_2$ sublattice thus contains an isolated chain of length 10. Consequently, this configuration
does not belong to $H_{Q_2}$.}
\label{fig:notinHQ2}
\end{center}
\end{figure}

A configuration does not belong to $H_{Q_2}$ if it contains one or more isolated open chains on the
sublattice $S_2$ with length $3p+1$. It is easy to see that all such configurations fall into the class of
configurations described in lemma \ref{lm:notinHQ2-a}, or lemma \ref{lm:notinHQ2-b}, or both. It follows that
all configurations that do not fall into either of these classes belong to $H_{Q_2}$.

\vspace{0.5cm}
\noindent
{\bf Step 2}

\vspace{0.1cm}
\noindent
As in the previous section, we will now compute $H_{12}=H_{Q_1}(H_{Q_2})$.
\begin{definition}
Define a row of $S_1$ sites as the set of $S_1$ sites directly above one $S_2$ chain.
\end{definition}

Note that the configurations in $H_{Q_2}$ again contain segments where $f_2$, the number of fermions on the
$S_2$ sublattice, is zero and segments where it is non-zero.

\begin{lemma}\label{lm:noS1occ}
Lemma \ref{lm:HQ1} for $H_{12}$ holds for each row of $S_1$ sites.
\end{lemma}

That is, in the segments where $f_2=0$, the cohomology of $Q_1$ vanishes when the number of $S_1$ sites
between any pair of bounding sites is $3p+1$ and it contains one element otherwise. The proof can be found in
the previous section. It follows that, in the
segments where $f_2=0$, two types of configurations on a row of $S_1$ sites are allowed. Using the notation of
the previous section, the two types can be distinguished by containing $3s-1$ dots or $3s$ dots.

\begin{lemma}\label{lm:segmentwidths}
The configurations in $H_{12}$ have spatially separated columnar
segments where $f_2=0$ and segments where $f_2>0$. The width of a
column in a segment where $f_2=0$ can vary between $3s+1$ and
$3s+2$ $S_1$ sites, whereas the width of a column in a segment
where $f_2>0$ can vary between $3p-1$ and $3p+1$ $S_1$ sites. In
the latter case, two consecutive rows never both have width $3p$
and the difference in their widths is at most 1 (or -1).
\end{lemma}
\begin{proof}
This follows from combining lemma's \ref{lm:notinHQ2-a}, \ref{lm:notinHQ2-b} and \ref{lm:noS1occ}. 
\end{proof}
An example is shown in figure \ref{fig:column}. From lemma \ref{lm:segmentwidths} it follows that we only have to
consider columns of width varying between 1 and 2 in the segments where $f_2=0$ separated by columns of width
varying between 2 and 4 in the segments where $f_2>0$. All other configurations can be obtained from these
configurations by inserting multiples of 3 dots in the segments where $f_2=0$ over the entire height of the
columns, and, similarly, by inserting multiples of 3 zeroes in the segments where $f_2>0$ over the entire
height of the columns.

\begin{figure}[h!]
\begin{center}
\includegraphics[width= .6\textwidth]{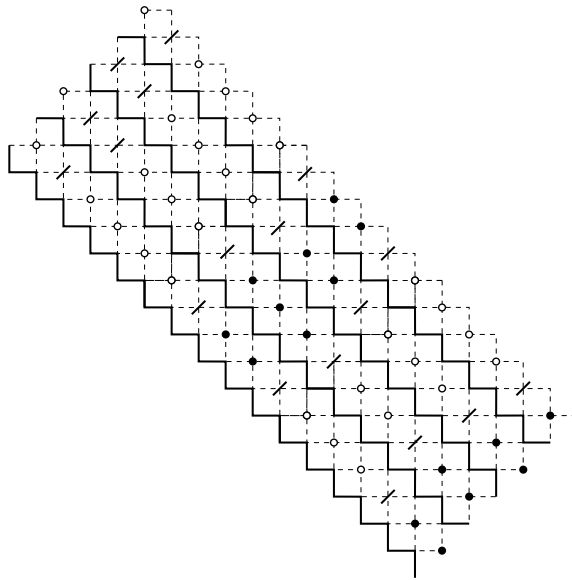}
\caption{Part of a configuration is shown.}
\label{fig:column}
\end{center}
\end{figure}

We now turn to the segments where $f_2>0$. Remember that in the previous section this step was easy
because all $S_1$ sites in the segment where $f_2>0$ were blocked by fermions on the $S_2$ chain. Here,
however, that is not the case. The first thing we note in this case is the following.

\begin{lemma}
The $S_1$ sites within a column marking a segment where $f_2>0$ have to be empty if they are away from the
boundaries with adjacent columns marking a segment where $f_2=0$.
\end{lemma}
\begin{proof}
This follows directly from lemma \ref{lm:notinHQ2-a}. 
\end{proof}

From this lemma it follows that we only have to consider the $S_1$ sites on the boundary between a segment
where $f_2>0$ and a segment where $f_2=0$. In fact, we will argue that we only have to
consider the boundary where the segment with $f_2>0$ is to the right of a segment with $f_2=0$ (and not the
boundary on the other side).

First, however, we introduce a new notation where a configuration
is fully characterized by the boundaries between the two types of
segments ($f_2=0$ and $f_2>0$). From lemma \ref{lm:notinHQ2-a} it
follows that these boundaries are an arbitrary sequence of steps
of $+(2,1)$ and $+(1,2)$. However, in the new notation we shall
tilt the lattice by $-45^{\circ}$, such that the rows of $S_1$
sites are horizontal. If we then draw the boundary as a collection
of vertical lines between two $S_1$ sites that are to the left and
to the right of the boundary, we find that the boundaries have a
zigzagged shape. The segments where $f_2>0$ will be white and the
segments where $f_2=0$ will be grey. For an example see figure
\ref{fig:boundary}.

\begin{figure}[h!]
\begin{center}
\includegraphics[width= .6\textwidth]{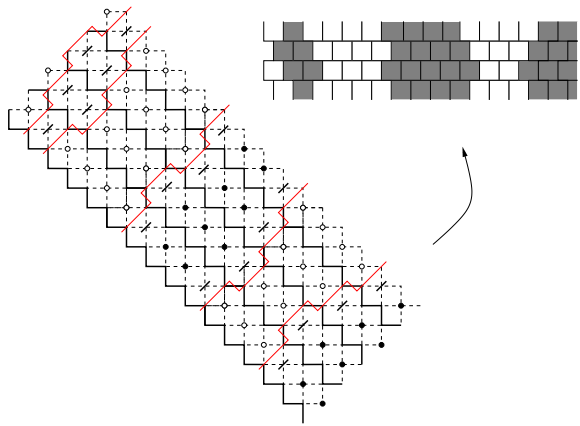}
\caption{Part of a configuration is shown with a mapping to the new notation.}
\label{fig:boundary}
\end{center}
\end{figure}

Suppose for a moment
that we would have a completely disconnected graph, that is, just a collection of disconnected vertices. Then
each site can be both empty and occupied. It is clear that each configuration with a least one
empty site does not belong to the kernel of $Q$, whereas each configuration with at least one occupied site
belongs to the image of $Q$. It follows that $H_Q$ vanishes at all grades. Here we do not have a disconnected
graph, however, it turns out that the division in grey and white regions is similar to disconnecting the
graph.

We define a special notation for a site that can be both empty and occupied. If this site is to the
right of a grey region we shall denote this site with a dot, whereas when it is to the left of a grey region
the site will be shaded. That is, suppose there are two configurations that both belong to $H_{Q_2}$ and
obey lemma \ref{lm:segmentwidths}, such that these two configurations differ by one site only. Then we can
summarize these two configurations in one picture by denoting this particular site by a
dot if it is to the right of a grey region or by shading the site if it is to the left of a grey region. For
an example see figure \ref{fig:dotexample}. Moreover, we can summarize $2^n$ configurations in one
picture if the picture contains $n$ sites with
dots or shaded sites. We make a distinction between sites to the left and to the right of the grey region,
because we will argue that the configurations with a site with a dot do not belong to $H_{12}$. Clearly this
is a choice, we could also have chosen to argue that the configurations with a shaded site do not belong to
$H_{12}$.

Let us consider a boundary that separates a grey segment on the left from a white segment on the
right. There are only a few such configurations that may have a site on the boundary with a dot.

\begin{figure}[h!]
\begin{center}
\includegraphics[width= .6\textwidth]{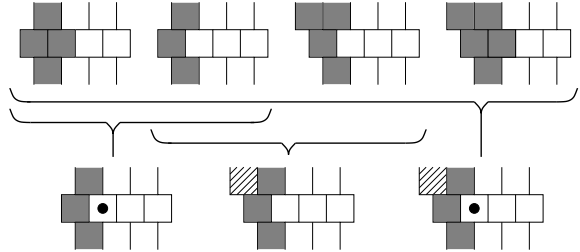}
\caption{The big brackets indicate how the configurations at the
top can be summarized using the notation introduced in the text.
The two left-most configurations at the top differ by one site in
the \emph{right}-most boundary of the grey region, they can
therefore be summarized by the left-most picture at the bottom by
denoting this site with a dot. The middle two configurations at
the top differ by one site in the \emph{left}-most boundary of the
grey region, they can therefore be summarized by the middle
picture at the bottom by shading this site. Finally, all four
configurations at the top can be summarized by the right-most
picture at the bottom.} \label{fig:dotexample}
\end{center}
\end{figure}

\begin{lemma}\label{possdot}
There are 8 possible configurations with a site with a dot in a
boundary that separates a grey segment on the left from a white
segment on the right. The configurations are depicted in figure
\ref{fig:dot3}.
\end{lemma}
\begin{proof}
The first restriction follows from lemma \ref{lm:notinHQ2-a}. That
is, both when the site is empty as well as when it is occupied,
the configuration should satisfy the lemma. This restriction is
depicted in figure \ref{fig:dot1}. Then the second restriction
follows from lemma \ref{lm:noS1occ}, that is, the width of a grey
column varies between 1 and 2 modulo 3. It follows that next to
the site with the dot there can only be one grey site (modulo 3).
Combining this again with lemma \ref{lm:notinHQ2-a}, we find four
possibilities for the left boundary of the grey segment. The four
possibilities can be summarized in one picture with the notation
defined above, see figure \ref{fig:dot2}, the shaded sites can be
both empty and occupied. Finally, it follows from lemma
\ref{lm:notinHQ2-b} that there are only two possible right most
boundaries for the white segment, each modulo columns of width 3,
see figure \ref{fig:dot3}. 
\end{proof}

\begin{figure}
     \centering
     \subfigure[]{
          \label{fig:dot1}
          \includegraphics[height=1.5cm]{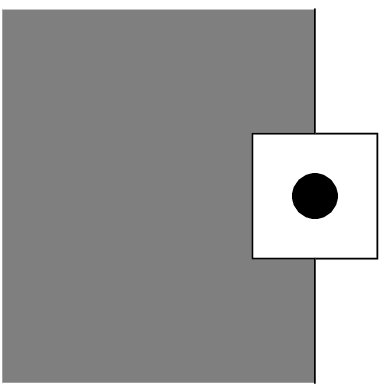}}
     \hspace{1cm}
     \subfigure[]{
          \label{fig:dot2}
          \includegraphics[height=1.5cm]{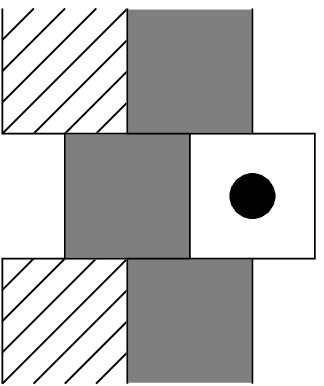}}
     \hspace{1cm}
     \subfigure[]{
           \label{fig:dot3}
           \includegraphics[height=1.5cm]{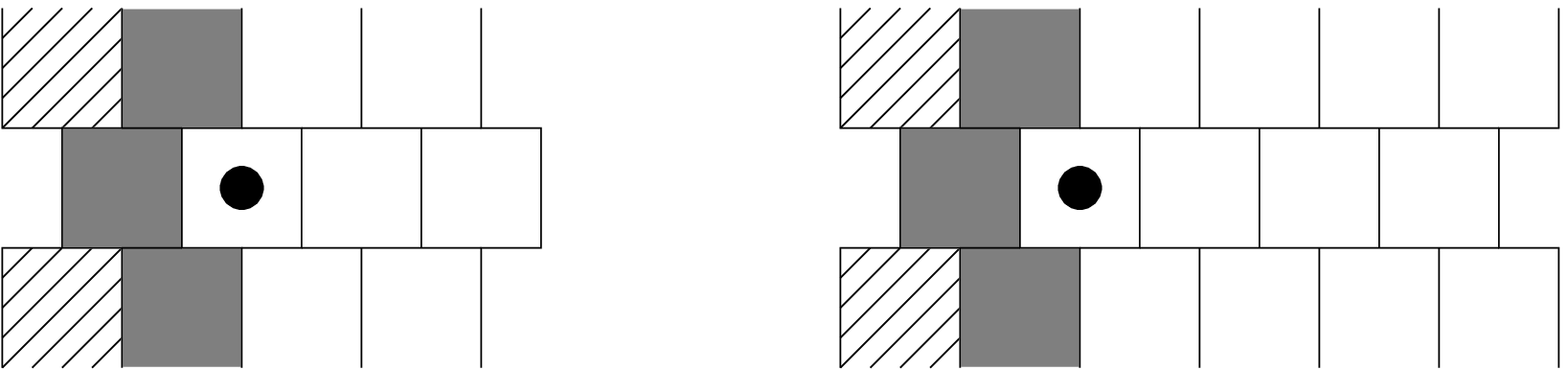}}
     \caption{In three steps we find that there are 8 possible configurations
     with a site with a dot in a boundary that separates a grey segment on the
     left from a white segment on the right: (a) It follows from lemma \ref{lm:notinHQ2-a}
     that a site with a dot must have occupied sites to the upper left and lower left and
     empty sites to the upper right and lower right. (b) There are four possibilities
     for the left-most boundary of the grey segment, following from the two
     shaded sites being empty or occupied. (c) There are two possibilities for the
     right-most boundary of the white segment.}
     \label{fig:dot}
\end{figure}

From lemma \ref{possdot} it follows that if there is more than one
site with a dot in the same boundary, they are sufficiently far
away to be independent. That is, each of these sites can be both
empty and occupied independent of the configuration of the other
dotted sites. Also note that if we select one of the 8
configurations with a dot from figure \ref{fig:dot3}, the rest of
the system can take on any configuration independent of the
configuration of the dotted site. Note that this resembles a
disconnected graph.

We are now ready to solve $H_{Q_1}(H_{Q_2})$.
\begin{lemma}\label{lm:notinHQ1}
All configurations that contain a boundary between a grey segment to the left and a white segment to the right,
such that this boundary contains one or more sites with a dot, do not belong to $H_{Q_1}(H_{Q_2})$.
\end{lemma}
\begin{proof}
A site with a dot can be either empty or occupied. Suppose the site is empty and we act with $Q_1$ on the
configuration. If $Q_1$ can act non-trivially only on the site under consideration we are done, since the
configuration in which the site with the dot is empty does not belong to the kernel of $Q_1$ and the
configuration in which it is occupied belongs to the image of $Q_1$. This proves the lemma for this case.

If, however, $Q_1$ can act non-trivially
also on other sites, there are four scenarios: a) The other site is in the same boundary. b) The other site is
in the left boundary of the grey region under consideration. c) The other site is in the right boundary of
the white region under consideration. d) The other site is further away from the region under consideration
than the first three cases.

For scenario a), we know that the other site is also a site with a dot. It follows that the configuration
with both dotted sites empty does not belong to the kernel of $Q_1$. The sum of the configurations with one
of the dotted sites empty belongs to the image of $Q_1$. The difference of the configurations with one
of the dotted sites empty does not belong to the kernel of $Q_1$, because it maps to the configuration with
both dotted sites occupied. Clearly, the latter configuration belongs to the image of $Q_1$. So for this
scenario the lemma is proven.

For scenario b) we distinguish two cases. First, the other site
and the dotted site can be occupied simultaneously. In this case
we can prove the lemma via the same argument as we did for
scenario a). Second, the other site and the dotted site {\it
cannot} be occupied simultaneously. This only happens when the
other site is in the same row as the dotted site. In this case the
sum of the configurations with one of them occupied is in the
image of $Q_1$, but the difference belongs to the kernel of $Q_1$
and does not belong to the image of $Q_1$. The latter is thus an
element of $H_{Q_1}(H_{Q_2})$. However, we have the freedom to
decide to keep only the configuration in which the other site is
occupied and the dotted site is empty as a representative of this
element. At this point it becomes clear why we only consider
configurations with a site with a dot, and not configurations with
a shaded site.

For scenario c) we can again distinguish two cases. In the first
case, the configuration of the other site and the dotted site are
independent and the lemma is proven as for scenario a). In the
second case, the other site and the dotted site {\it cannot} be
occupied simultaneously. There are again three configurations
under consideration. The configuration with both sites empty does
not belong to $\ker Q_1$, the sum of the configurations with one
of the two sites occupied belongs to $\im Q_1$ and, finally, the
difference again is an element of $H_{Q_1}(H_{Q_2})$. And as under
b), we choose to represent this element with the configuration
where the dotted site is empty and the other site occupied.

Finally, for scenario d) it is clear that the configuration of the other site and the dotted site are always
independent and the lemma is proven as for scenario a).

In the four scenarios we considered, there was just one other site on which $Q_1$ acts non-trivially. If
there are more sites on which $Q_1$ acts non-trivially, the lemma clearly holds when these sites can again be
empty or occupied independent of the dotted site. However, if they are not all independent, the proof is more
lengthy, but analogous to the proofs of the second case in scenarios b) and c).

\end{proof}

\begin{lemma}
All the configurations that belong to $H_{Q_1}(H_{Q_2})$ are a sequence of alternating grey and white
columns subject to the conditions in lemma \ref{lm:segmentwidths}, such that the left-most boundary of all
the white columns does not contain any sites with a dot.
\end{lemma}
\begin{proof}
This is a direct consequence of lemma \ref{lm:notinHQ1}. 
\end{proof}

As an example we consider the case where $\vec{v}=(6,6)$ and
$\vec{u}=(m,-m)$, that is, we stack four rows of $m$ $S_1$ sites
separated by four $S_2$ chains. All configurations in $H_{12}$ can
be obtained by concatenating the configurations depicted in figure
\ref{fig:example6komma6}, with eventual insertions of grey and/or
white columns of width 3, such that the boundary conditions are
satisfied\footnote{These configurations are obtained as follows.
First consider all possible white segments satisfying the boundary
condition in the $\vec{v}$-direction, then construct all
possibilities for the grey segments to the left of these white
segments.}. That is, each row should in the end have width $m$ or,
equivalently, the right-most boundary should fit with the
left-most boundary. Finally, there is a configuration with all
zeroes (one entirely white segment) and a configuration with all
dots (one entirely grey segment).

\begin{figure}[h!]
\begin{center}
\includegraphics[width= .9\textwidth]{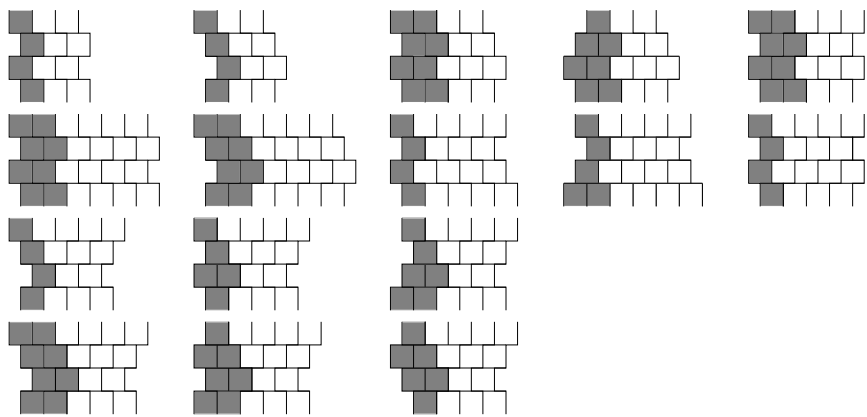}
\caption{Building blocks of the configurations spanning $H_{Q_1}(H_{Q_2})$ for $\vec{v}=(6,6)$ and
$\vec{u}=(m,-m)$, modulo insertions of grey and white columns of width 3.}
\label{fig:example6komma6}
\end{center}
\end{figure}

\vspace{0.5cm}
\noindent
{\bf Step 3}

\vspace{0.1cm}
\noindent
In the previous step we have determined $H_{12}$. According to the
'tic-tac-toe' lemma, the cohomology of $Q$ is equal to or
contained in $H_{12}$: $H_Q \subseteq H_{12}$. In the previous
section we found that for $\vec{v}=(1,2)$, we have $H_Q = H_{12}$.
For general $\vec{v}$, however, this is not the case. That is,
within $H_{12}$, there are configurations that are not in the
kernel of $Q$ and there are configurations that are in the image
of $Q$. To find out which configurations do not belong to $H_Q$,
we follow the 'tic-tac-toe' procedure \cite{botttu} as described
in step 3 of section \ref{sec:zz}.

In the previous section, we found via the 'tic-tac-toe' procedure
that we could find for each element $\ket{\psi_0}$, that belongs
to $H_{12}$, but not to $\ker Q$, an element $\ket{\psi_n}$ that
does belong to $\ker Q$. In this section, however, we will find
that for some elements $\ket{\psi_0}$, the 'tic-tac-toe' procedure
leads to a corresponding element $\ket{\tilde{\psi}}$, that also
belongs to $H_{12}$. We then say that $\ket{\psi_0}$ maps to
$\ket{\tilde{\psi}}$ at the end of the 'tic-tac-toe' procedure and
we conclude that neither $\ket{\psi_0}$ nor $\ket{\tilde{\psi}}$
belong to $H_Q$.

We now prove some rules for the 'tic-tac-toe' procedure specific to the configurations we obtained in the
previous step.

\begin{lemma}\label{ttt1}
Let $Q$ act on an empty $S_1$ site $(k,l)$, such that for the preceding $S_1$ sites on that row we have:
$(k-1,l+1)$ and $(k-3s-2,l+3s+2)$ are occupied and the intermediate sites are dotted. Then the new
configuration with $(k,l)$ occupied, is also the image of $Q_1$ acting on the configuration with $(k,l)$
occupied and one less fermion in the preceding sites $(k-1,l+1)$ to $(k-3s-1,l+3s+1)$.
\end{lemma}
\begin{proof}
For general $s$ we can denote the original configuration as "$1 \cdot_{3s} 10$", the new configuration is then "$1
\cdot_{3s} 11$". From lemma \ref{lm:HQ1} we know that, if the number of $S_1$ sites between a bounding pair
is $3s+1$, $H_{Q_1}$ vanishes. Consequently, each configuration that is in the kernel of $Q_1$ is also in the
image of $Q_1$. Now, since the configuration "$1\cdot_{3s} 11$" is in the kernel of $Q_1$ and the number of
$S_1$ sites between the bounding pair is $3s+1$, it must also be in the image of $Q_1$. Thus, there is a
configuration with one less fermion between the bounding pair that maps to this configuration under the
action of $Q_1$. 
\end{proof}
\begin{example}
For $s=0$ this is easily understood: the original configuration will have "$110$" on the $S_1$ sites
$(k-2,l+2)$ through $(k,l)$. Acting on this with $Q$ gives "$111$", however, this can also be obtained by
acting with $Q$ on -"$101$".
\end{example}

\begin{lemma}\label{tttnew1}
Acting with $Q$ on a white segment away from the boundary, gives zero.
\end{lemma}
\begin{proof}
The proof is analogous to the proof of lemma
\ref{Q1zerowithinHQ2}. The length of the $S_2$ chains in the white
region is $L=3k$ or $L=3k-1$ each containing $k$ fermions. If
$Q_1$ acts on a site above this chain and away from the boundary,
it will cut the $S_2$ chain into 3 pieces. One of length 1 and two
of lengths $L_1'$ and $L_2'$, such that $L_1'+L_2'=L-3$. We will
now argue that the new configuration with the smaller $S_2$
chains, always belongs to $\im Q_2$. This implies that we can
always continue to the next step in the 'tic-tac-toe' procedure.

If the chain of length 1 contains a fermion, the new configuration
clearly belongs to $\im Q_2$. If it is empty there are $k$
fermions on the other two chains. For $L=3k$ their combined length
is $L_1'+L_2'=3(k-1)$, so $L_1'=3k_1$ and $L_2'=3k_2$ or
$L_1'=3k_1+1$ and $L_2'=3k_2-1$, where in both cases
$k_1+k_2=k-1$. For the second case the cohomology vanishes for all
fermion numbers because of the length $L_1'$. For the first case
the cohomology is non-vanishing only if $f=k_1+k_2=k-1$, however,
there are $k$ fermions. So for both cases the new configuration
belongs to $\im Q_2$ (since it belongs to $\ker Q_2$ and not to
$H_{Q_2}$). For $L=3k-1$ we find $L_1'=3k_1$ and $L_2'=3k_2-1$ or
$L_1'=3k_1+1$ and $L_2'=3k_2-2$, where in both cases
$k_1+k_2=k-1$. The rest of the argument is the same as before.

From the above it follows that we can always continue with the next step in the 'tic-tac-toe' procedure. Now
suppose that in this next step we act with $Q_1$ on the same row as in the first step. Since there are no
fermions between these two $S_1$ sites, this configuration will cancel against the configuration where the
two $S_1$ sites are occupied in the reverse order. It follows that we only have to consider acting with $Q_1$
on each row just once.

It is now easily verified that, since there are as many $S_2$ chains as there are $S_1$ rows, we can always
continue the 'tic-tac-toe' procedure until we get zero. 
\end{proof}

In this lemma we restricted ourselves to $Q$ acting on $S_1$ sites away from the boundary. We will see in the
following that if we allow $Q$ to act on sites at the boundary, the 'tic-tac-toe' procedure can map
one configuration in $H_{12}$ to another configuration in $H_{12}$. The crucial point is that, when we act
with $Q_1$ on a site at the boundary, the length of at least one of the $S_2$ chains below and above this
site is reduced by 1. If the original length was $3k$, the new length is $3k-1$ and both have
non-vanishing cohomology for $f=k$. In that case we cannot use this $S_2$ chain to write the new
configuration as $Q_2$ of some other configuration. It follows that to continue the 'tic-tac-toe' procedure,
we have to use the other $S_2$ chain. However, if this chain was already used in a previous step, the
'tic-tac-toe' procedure could end. Before we continue with an example that illustrates this point, we will
argue that it is enough to consider $Q$ acting only on sites at the boundary. This follows from lemma
\ref{lm:segmentwidths}; if the 'tic-tac-toe' procedure ends because we have obtained a configuration that
does not belong to $\im Q_2$ (nor to $\im Q_1$), this configuration must belong to $H_{12}$. From lemma
\ref{lm:segmentwidths} we know that configurations in $H_{12}$ have spatially separated columnar grey and
white segments that do not branch. It follows that we can only map one configuration in $H_{12}$ to another
by either creating a new grey column in a white column, or by (locally) increasing the width of a grey
column. Since the first possibility is excluded by lemma \ref{tttnew1}, we conclude that we can restrict $Q$
to act only on sites at the boundary. As in step 2 we will restrict
ourselves to the left-most boundary to avoid overcounting.

Let us consider an example of a configuration that does belong to $H_{12}$, but does not belong to $H_Q$,
i.e. it maps to another configuration in $H_{12}$ at the end of the 'tic-tac-toe' procedure.

\begin{example}\label{ex:ttt}
Consider the configuration shown on the left in figure \ref{fig:tttexample}. We label the three $S_2$ chains (not shown
explicitly) between the four $S_1$ rows; chain 1, chain 2 and chain 3 ($c_1$, $c_2$ and $c_3$) from top to
bottom. Similarly, we label the $S_1$ rows; row 1 to row 4 ($r_1$ to $r_4$) from top to bottom. The $S_2$
chains have lengths $L_{c_1}=6$, $L_{c_2}=5$ and $L_{c_3}=3$ and thus contain 2, 2 and 1 particle respectively.
Now consider the left-most, empty $S_1$ sites in the middle two rows. Occupying the left-most, empty $S_1$
site on row 2 reduces the length of $c_1$ from 6 to 5. There will still be 2 particles on $c_1$ and since the
chain of length 5 has non-vanishing cohomology at grade 2, the configuration on this chain will in general
not belong to $\im Q_2$. Occupying the left-most, empty $S_1$ site on row 3 reduces the length of $c_3$ from
3 to 2. Again the configuration on this chain will not belong to $\im Q_2$, since the chain of length 2 has
non-vanishing cohomology at grade 1. It follows that if we occupy either of these $S_1$ sites in the
'tic-tac-toe' procedure, we have to use $c_2$ to write the new configuration as $Q_2$ of some other
configuration. By definition this is always possible in the first step of the procedure. However, also by
definition, we can do this only once since $Q^2=0$. It follows that, after two steps in the 'tic-tac-toe'
procedure, we obtain a new configuration (see fig. \ref{fig:tttexample} on the right) that has 2, 1 and 1 particles on the
$S_2$ chains from top to bottom and belongs to $H_{12}$. Consequently, both the original as well as the final 
configuration do not belong to the cohomology of $Q$, although they do belong to $H_{12}$.
\end{example}

\begin{figure}[h!]
\begin{center}
\includegraphics[width= .4\textwidth]{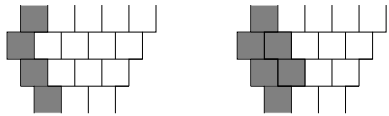}
\caption{On the left we depict the configuration that does belong to $H_{12}$, but not to $H_Q$, since it 
does not belong to the kernel of $Q$. Instead it maps to the configuration depicted on the right under the 
'tic-tac-toe' procedure. This configuration also belongs to $H_{12}$, but not to $H_Q$, since it belongs to $\im Q$. 
The configuration on the left has 4 particles on sublattice $S_1$ and 5 particles on sublattice $S_2$, divided as 
2, 2, 1 over the $S_2$ chains from top to bottom. The configuration on the right has one more particle in total; 
it has 6 particles on $S_1$ and it has 4 particles on $S_2$, divided as
2, 1, 1 over the $S_2$ chains from top to bottom.}
\label{fig:tttexample}
\end{center}
\end{figure}

As we anticipated, the crucial point in this example is that the length of an $S_2$ chain is reduced
from $3k$ to $3k-1$, since this limits the options to continue the 'tic-tac-toe' procedure. In fact, in the
'tic-tac-toe' procedure, we can only reach a configuration that is not in the image of $Q_2$ if the length
of an $S_2$ chain is reduced from $3k$ to $3k-1$. To find the configurations in $H_{12}$ that map to another
configuration in $H_{12}$ under the action of $Q$ in the most
efficient way\footnote{By 'the most efficient way' we mean the shortest sequence of occupying $S_1$ sites in 
the 'tic-tac-toe' procedure that maps one configuration in $H_{12}$ to another. This is the most efficient
way, because as soon as this happens, we know that both configurations do not belong to $H_Q$, independent 
of all the other terms created under the action of $Q$.}, we will start the 'tic-tac-toe' procedure by
occupying an $S_1$ site, such that this happens. It follows that
we can then only use the other $S_2$ chain to continue the
procedure. We will then, again for efficiency, continue the
procedure by again occupying an $S_1$ site such that there is just one $S_2$ chain that we
can use to continue the procedure. This means that we will act
with $Q_1$ on consecutive rows, either moving upwards or
downwards along the boundary.

In the previous step we constructed all possible configurations with a site {\em with a dot} in the left-most
boundary of a white segment. Here we will construct all possible configurations with a site in the left-most
boundary of a white segment, such that occupying this site reduces the length of an $S_2$ chain from $3k$ to
$3k-1$. We shall call such sites 'critical reducer sites'. We start with the white segment and obtain
the configurations depicted in figure \ref{fig:unblocked1}. For these configurations occupying the left-most
site of the middle row reduces the length of at least one of the adjacent $S_2$ chains from $3k$ to
$3k-1$. For the two configurations on the left, occupying this site reduces the
length of both $S_2$ chains to $3k-1$. It follows that the new configuration must belong to $\im Q_1$ (see lemma
\ref{ttt1}), otherwise it was a site with a dot in the previous step. So we
do not have to consider these two configurations. This same reasoning tells us that the grey region to the left of
the middle row should have width 1 modulo 3, otherwise the new configuration would belong to $\im Q_1$. This
leads to the 12 possibilities in figure \ref{fig:notinkerQ}.

\begin{figure}[h!]
\begin{center}
\includegraphics[width= .5\textwidth]{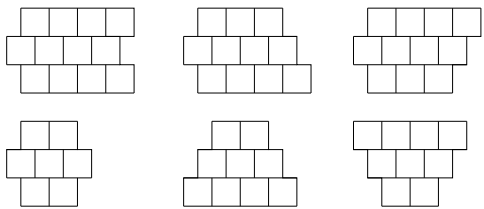}
\caption{The possible boundaries of a white segment are shown, such that the left-most boundary contains a critical 
reducer site. That is, occupying this site reduces the length of at least one of the adjacent $S_2$ chains from $3k$ to $3k-1$.}\label{fig:unblocked1}
\end{center}
\end{figure}

\begin{figure}[h!]
\begin{center}
\includegraphics[width= .5\textwidth]{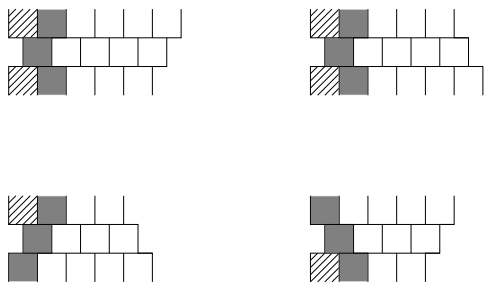}
\caption{The 12 possible configurations such that the left-most site of the middle row is a critical reducer and
occupying this site does not lead to a configuration that is in $\im Q_1$.}
\label{fig:notinkerQ}
\end{center}
\end{figure}

Note that, indeed, occupying the critical reducer site at the boundary, leads to reducing the length of one of the $S_2$
chains from $3k$ to $3k-1$, for all these configurations. For efficiency we continue the 'tic-tac-toe' procedure either 
upwards or downwards, such that at every step in the procedure there is just one $S_2$ chain that we can use to continue 
the procedure. The direction we should follow, is indicated by the arrow in figure \ref{fig:tttdirection}. Note that
we dropped the two configurations for which the grey segment had width 2, because of lemma \ref{ttt1}.

\begin{figure}[h!]
\begin{center}
\includegraphics[width= .5\textwidth]{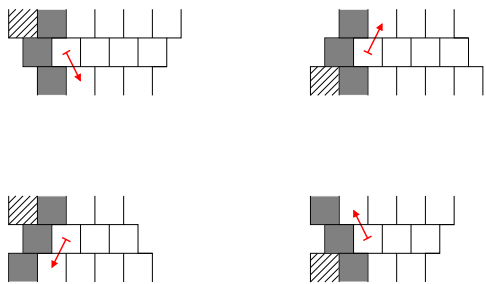}
\caption{If the critical reducer site (the left-most site of the middle row) is occupied in the first step of the
'tic-tac-toe' procedure, the procedure should be continued in one direction only, as explained in the text. This
direction is indicated by the arrow.}\label{fig:tttdirection}
\end{center}
\end{figure}

It is now clear that if we stack a configuration for which the
'tic-tac-toe' procedure goes downwards on top of a configuration
for which it goes upwards, the 'tic-tac-toe' procedure ends. In
particular, it maps the old configuration to a new configuration
that is also in $H_{12}$. We can increase the number of steps
necessary in the 'tic-tac-toe' procedure by stacking rows for
which the grey segment has width 1 modulo 3 and the width of the
white segment alternates between 3 and 4 modulo 3 (see fig.
\ref{fig:tttlonger}). Examples of the stacked configurations and
the configurations they map to are shown in figure
\ref{fig:tttstack}. Here the sites with connected dots can be
either all empty or all occupied. The configuration with all the
sites empty maps to the configuration with all the sites occupied
under the 'tic-tac-toe' procedure. However, if the configurations
on the left in figure \ref{fig:tttdirection} are not combined with
one of the configurations on the right in figure
\ref{fig:tttdirection}, the 'tic-tac-toe' procedure will end with
a state that is in the kernel of $Q$ (as long as we only let the
sites on the left-most boundary participate).

\begin{figure}[h!]
\begin{center}
\includegraphics[width= .4\textwidth]{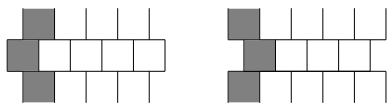}
\caption{Stacking these configurations with the configurations of figure \ref{fig:tttdirection}, increases the
number of steps in the 'tic-tac-toe' procedure.}\label{fig:tttlonger}
\end{center}
\end{figure}

\begin{figure}[h!]
\begin{center}
\includegraphics[width= .7\textwidth]{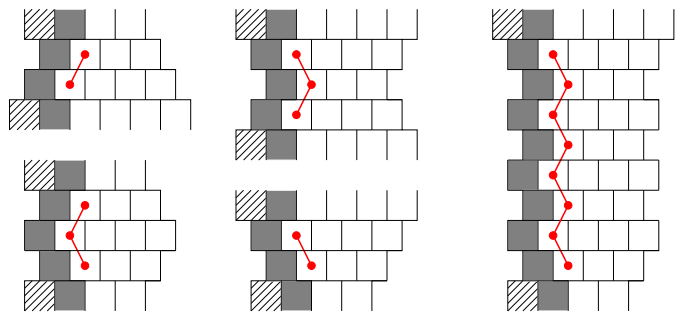}
\caption{Some examples are shown of configurations that do belong to $H_{12}$, but not belong to $H_{Q}$. Here
the sites with connected dots can be either all empty or all occupied.}\label{fig:tttstack}
\end{center}
\end{figure}

At this point we have identified a certain set of configurations that does belong to $H_{12}$, but
does not belong to $H_{Q}$. However, we have to make a final step before we can identify all configurations
in $H_{Q}$ with tiling configurations. This is due to the fact that certain parts of configurations seem to
belong to $H_{Q}$, but they do not respect the boundary conditions. Note that, at this point, we have reduced
all possible motifs to the following set:\\
"$100$"\\
"$1100$"\\
"$10000$"\\
"$110000$"\\
which can be separated by single insertions of the motifs:\\
"$1000$"\\
"$11000$"\\
all modulo insertions of three dots and three zeroes along an entire column. Each
of the four basis motifs, comes with two directions, determined by whether the boundaries between the grey
and white segments follows the direction $(-1,-2)$ or $(-2,-1)$.

\begin{definition}
We assign a letter to each of the four basis motifs:\\
$A_i \equiv$"$100$"\\
$B_i \equiv$"$1100$"\\
$C_i \equiv$"$10000$"\\
$D_i \equiv$"$110000$"\\
where the subscript $i$ is 1 or 2 when the direction of the motif is $(-1,-2)$ or $(-2,-1)$ respectively.
\end{definition}

Note that the direction of a motif is not defined if neither the motif directly above it nor the motif
directly below it is the same. We will start, however, by considering cases in which this does not happen. At
the end of step 4 we will encounter a case where this point needs some attention.

We now want to study whether a vertical sequence of a certain motif can be followed by a sequence of a
different motif, eventually, with a insertion of one of the motifs with 3 zeroes: "$1000$" or "$11000$".

\begin{example}
As an example let us start with a sequence of motif $A_1$. This
sequence could be followed by motifs $A_2$, $B_1$ and $C_2$.
However, it cannot be followed by motif $B_2$, because it would
not belong to $H_{12}$. Nor can it be followed by motifs
$C_1$ or $D_i$, because it would not belong to $H_Q$ (see fig.
\ref{fig:A1toall}).
\end{example}

\begin{figure}[h!]
\begin{center}
\includegraphics[width= .6\textwidth]{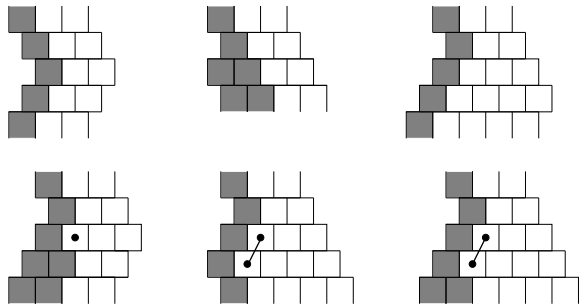}
\caption{At the top, we show, from left to right, motif $A_1$ followed by the motifs $A_2$, $B_1$ and $C_2$.
On the bottom-left, we see that a configuration in which $A_1$ is followed by $B_2$ contains a site with a
dot in the left-most boundary of the white segment. The other two configurations on the bottom, show that
configurations in which $A_1$ followed by $C_1$ or $D_i$ do not belong to $H_Q$. Here we used the notation of
figure \ref{fig:tttstack}.}\label{fig:A1toall}
\end{center}
\end{figure}

Similarly we find the following:
\begin{itemize}
\item{motif $B_1$ can only be followed by motif $C_2$.}
\item{motif $C_2$ can only be followed by motif $B_1$.}
\item{motif $A_i$ can be followed by motifs $A_j$, $B_1$ and
$C_2$.}
\item{motif $D_i$ can be followed by motifs $D_j$, $B_1$
and $C_2$.}
\item{motif $B_2$ can only follow after motif $C_1$.}
\item{motif $C_1$ can only follow after motif $B_2$.}
\end{itemize}

Finally, we know from lemma's \ref{lm:notinHQ2-b} and \ref{lm:noS1occ} that grey and white columns cannot
branch or have end points, since their width always oscillates between 1 and 2 modulo 3 or 2, 3 and 4 modulo 3
for the grey and white segments respectively. Consequently, grey and white columns may wind around the torus
several times, but they will always close to form a loop.

Let us combine this observation with the rules we found for stacking motifs. Consider, for example, motif
$A_1$, which can be followed by motifs $A_2$, $B_1$ and $C_2$. However, motifs $B_1$ and $C_2$ can only be
followed by $C_2$ and $B_1$ respectively. Consequently, if motif $A_1$ is followed by either of these two
motifs, we can never fulfill the boundary conditions, because the column cannot be closed to form a loop.
Thus configurations in which motif $A_1$ is followed by motifs $B_1$ and $C_2$ do not belong to $H_Q$. In
this same spirit we obtain the following lemma.

\begin{lemma}
For configurations that belong to $H_Q$ the following holds:
\begin{itemize}
\item{motif $B_1$ can only be followed by motif $C_2$ and vice versa.}
\item{motif $A_1$ can only be followed by motif $A_2$ and vice versa.}
\item{motif $D_1$ can only be followed by motif $D_2$ and vice versa.}
\item{motif $B_2$ can only be followed by motif $C_1$ and vice versa.}
\end{itemize}
\end{lemma}

For motifs $A_i$ and $D_i$ the width of the white segment does not
change, thus the motif with direction 1 can follow directly below
or above this same motif with direction 2. For the motifs $B_i$
and $C_i$, however, there is an intermediate motif of the type
"$1000$" or "$11000$". Which of the two can be determined via the
'tic-tac-toe' procedure. If we read the motifs of the rows from
top to bottom, we find that a sequence of $B_1$ motifs will be
followed by "$1000$", to be followed by a sequence of $C_2$
motifs. Then the $C_2$ motifs will be followed by "$11000$", which
is then to be followed by another sequence of $B_1$ motifs. On the
other hand, a sequence of $B_2$ motifs will be followed by
"$11000$", followed directly by a sequence of $C_1$ motifs.
Finally, the $C_1$ motifs will be followed by "$1000$", followed
directly by another sequence of $B_2$ motifs (see fig.
\ref{fig:cornersBC}). It is readily checked that any other choice
gives a configuration that does belong to $H_{12}$, but not to
$H_Q$.

\vspace{0.5cm}

\begin{figure}[h!]
\begin{center}
\psfrag{B1}{$B_1$}
\psfrag{C1}{$C_1$}
\psfrag{B2}{$B_2$}
\psfrag{C2}{$C_2$}
\includegraphics[width= .6\textwidth]{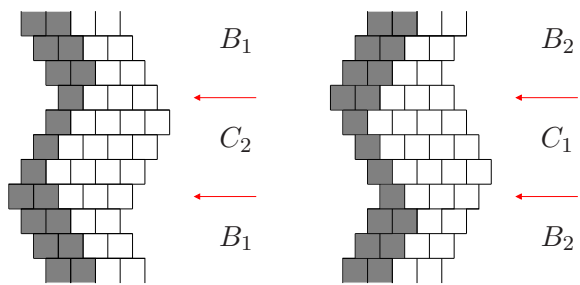}
\caption{Two configurations are shown in which motif $B_i$ is followed by motif $C_j$ (where $i\neq j$) and vice versa, with
the correct intermissions of the motifs "$1000$" and "$11000$" (indicated by the arrows).}\label{fig:cornersBC}
\end{center}
\end{figure}

\newpage
\vspace{0.5cm}
\noindent
{\bf Step 4}

\vspace{0.1cm}
\noindent
We are now ready to make the identification with the tiles. For the four basis motifs $A_1$ through $D_1$
the identification is shown in figure \ref{fig:basis4dir1} and $A_2$ through $D_2$ are identified with a
tiling in figure \ref{fig:basis4dir2}. Note that to distinguish motif $X_1$
from $X_2$, where $X=A,B,C$ or $D$, one has to consider also the motif on the row above or below this
motif. These motifs can be followed by an arbitrary threefold of zeroes. Let us define the motif $E
\equiv$"$000$". For this motif we can also distinguish a direction, because its boundary will follow the
left-most boundary of the white segment it is attached to. From figures \ref{fig:basis4dir1} and
\ref{fig:basis4dir2} it is clear that the motifs $X_i$ can be followed by motif $E_i$, where the $i$
should be the same. For the motifs $B_1$ and $C_2$ there is an exception: when the motif above these
motifs is "$11000$" and "$1000$" respectively, they are followed by $E_2$ and $E_1$ respectively.

\begin{figure}[h!]
\begin{center}
\includegraphics[width= .9\textwidth]{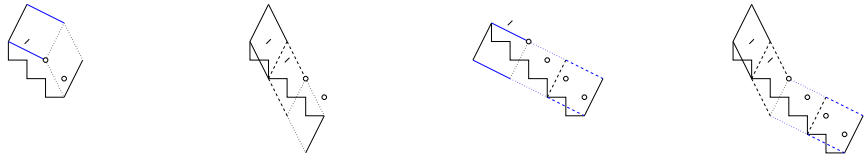}
\caption{The motifs $X_1$ and the corresponding tilings. Note that
this mapping was already found in section \ref{sec:zz} (see fig.
\ref{fig:tilingmapzigzag}).}\label{fig:basis4dir1}
\end{center}
\end{figure}

\begin{figure}[h!]
\begin{center}
\includegraphics[width= .9\textwidth]{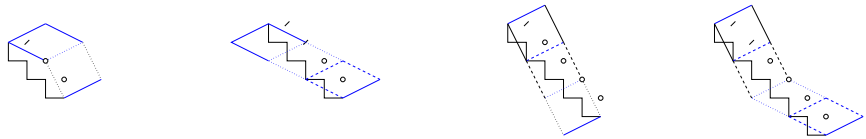}
\caption{The motifs $X_2$ and the corresponding tilings.}\label{fig:basis4dir2}
\end{center}
\end{figure}

\begin{figure}[h!]
\begin{center}
\includegraphics[width= .9\textwidth]{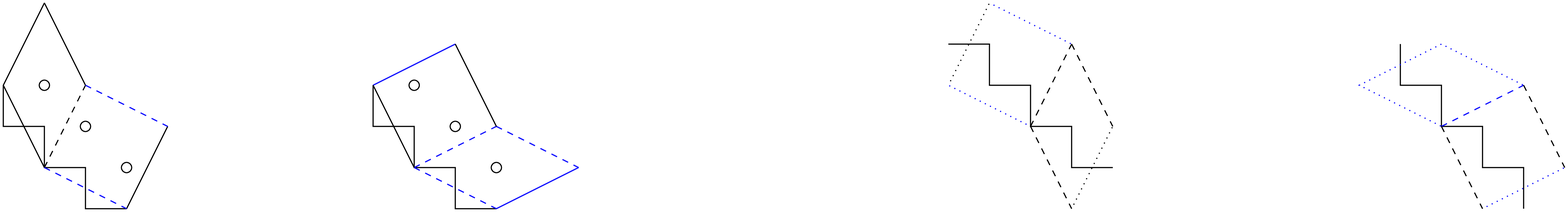}
\caption{On the left the motifs $E_1$ (on the left) and $E_2$ (on the right) and the corresponding tilings.
On the right the tilings corresponding to insertions of 3 dots into the motifs $X_1$ (on the left) and
$X_2$ (on the right).}\label{fig:zeroesdots}
\end{center}
\end{figure}

There can also be insertions of multiples of three dots in the four basis motifs. How this translates into
tilings is shown in figure \ref{fig:basis4dots}. More precisely, note that each basis motif $X_1$ contains
a dotted line, connecting two $S_1$ sites along the direction 1 (black) and equivalently, all basis motifs $X_2$ contain
a dotted line with direction 2 (blue). An insertion of three dots in a basis motif corresponds to an insertion of
two tiles, shown in figure \ref{fig:zeroesdots}, at this dotted line. Note that here we cannot easily write the
motif of dots directly in the tiles, however, the mapping is still unambiguous.

\begin{figure}[h!]
\begin{center}
\includegraphics[width= .9\textwidth]{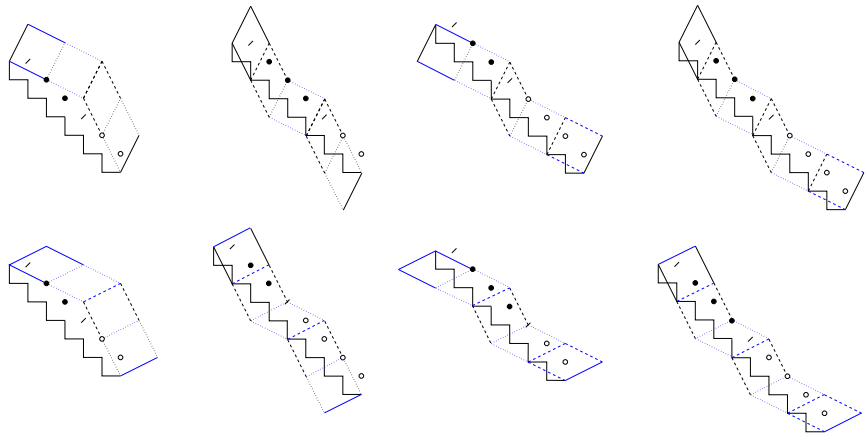}
\caption{On the top (bottom), insertions of 3 dots into the motifs $X_1$ ($X_2$) and the corresponding
tilings are shown.}\label{fig:basis4dots}
\end{center}
\end{figure}

Finally, we have the motifs "$1000$" and "$11000$". Which tilings these motifs correspond to depends on whether they
occur between the motifs $B_1$ and $C_2$ or the motifs $B_2$ and $C_1$. In fact, in the first case,
"$1000$" will correspond to the same sequence of tiles as $B_1$ and "$11000$" to the same tiling as $C_2$. Similarly, in
the latter case, "$1000$" and $B_2$, on the one hand, and "$11000$" and $C_1$, on the other hand, correspond
to the same tilings. For an example see figure \ref{fig:mapcornersBC}.

\begin{figure}[h!]
\begin{center}
\includegraphics[width= .7\textwidth]{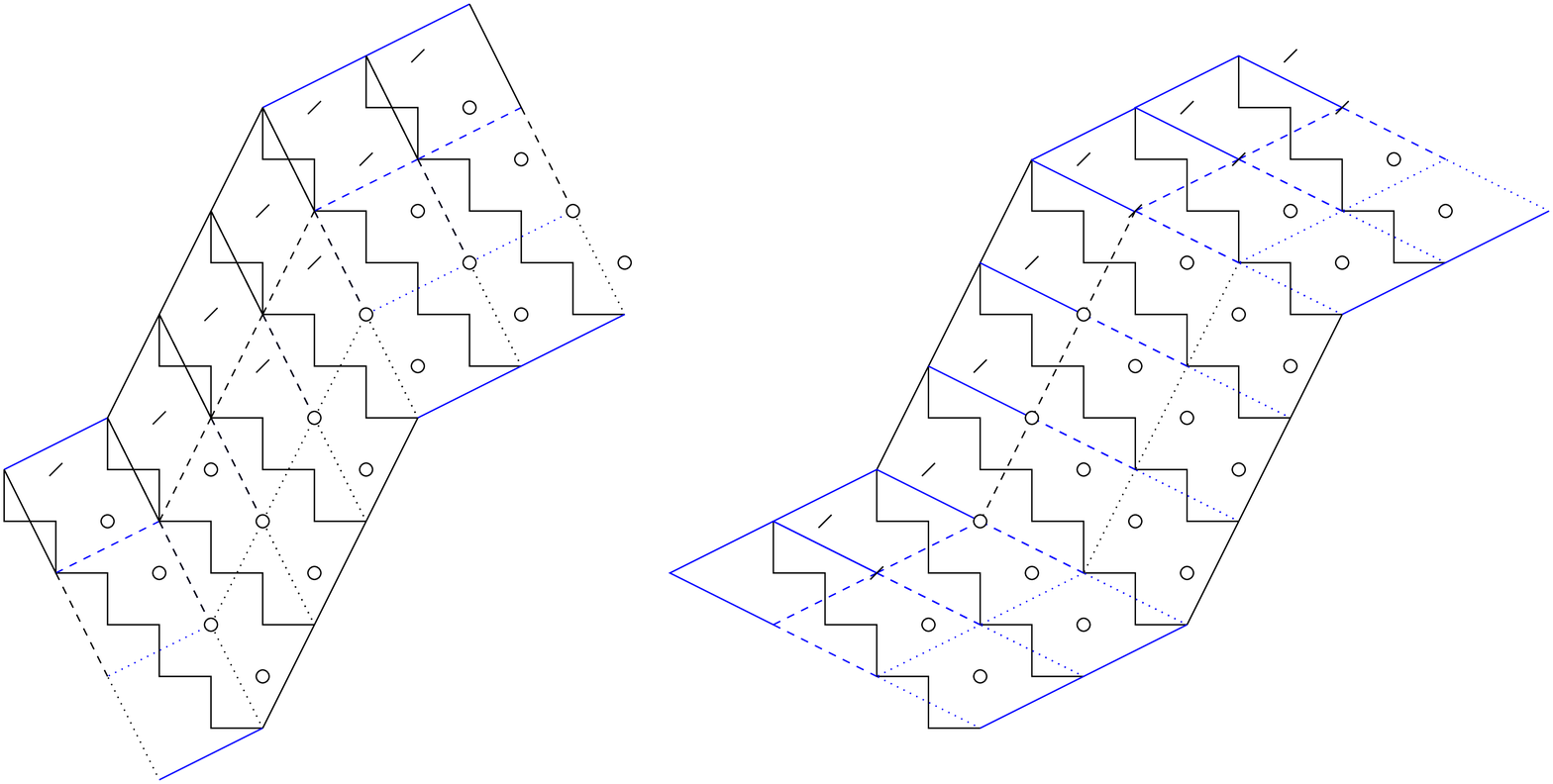}
\caption{On the left (right), we show a configuration in which motifs $C_2$ and $B_1$ ($C_1$ and $B_2$)
alternate with the corresponding tilings. Note that the motifs "$1000$" and "$11000$" correspond to different
tilings on the left than on the right.}\label{fig:mapcornersBC}
\end{center}
\end{figure}

With these identifications there is one ambiguity, but it is easily dealt with. If there is a column in which
the motifs alternate indefinitely between "$1100$" and "$11000$", we cannot determine whether the motif
"$1100$" is of type $B_1$ or $B_2$. The same happens when the motifs "$10000$" and "$1000$" alternate
indefinitely: the motif "$10000$" could be of type $C_1$ or $C_2$. Note that if we choose to identify the
first with $B_1$ and the second with $C_2$, the corresponding tilings would be indistinguishable. This also
happens when we choose $B_2$ and $C_1$. So we conclude that we should either choose $B_1$ and $C_1$ or $B_2$
and $C_2$. The ambiguity is thus lifted by simply deciding that we will always choose, say, $B_1$ and $C_1$.

Finally, we note that again the number of fermions in a certain configuration is the same as the number of
tiles in the corresponding tiling. For the four basis motifs and the motifs with 3 zeroes or 3 dots, this
follows from the arguments in section \ref{sec:zz}. For the motifs "$1000$" and "$11000$", we should look at
figure \ref{fig:mapcornersBC}. If the motif "$1000$" sits between motifs $B_1$ and $C_2$, the number of
fermions in these three rows is $3*3$ and the number of sites is 3 times the number of $S_1$ sites:
$3*(2*4+5)$. So 9 fermions on 39 sites. Compare this
with the corresponding tiling: it contains 2 times 3 tiles of area 4 and once 3 tiles of area 5. So 9 tiles
with total area 39. Similarly, if the motif "$11000$" sits between motifs $B_1$
and $C_2$, the number of fermions in these three rows is $3*3$ and the number of sites is $3*(2*5+4)$. The
corresponding tiling contains 2 times 3 tiles of area 5 and once 3 tiles of area 4. For the corners between motifs $B_2$
and $C_1$ the comparison is slightly more subtle. Following the same arguments as above, we find that in
this case the number of fermions in the motifs "$1000$"
and "$11000$" do not agree with the number tiles in the corresponding tiling. However, the discrepancy is
minus one in one case and plus on in the other, and since the boundary conditions dictate that the number of
"$1000$"-motifs equals the number of "$11000$"-motifs, the discrepancies exactly cancel.

\vspace{0.5cm}
\noindent
{\bf Step 5}

\vspace{0.1cm}
\noindent
The final step concerns the small correction $\Delta$ in equation
(\ref{eq:HQ}). With the four basis motifs, horizontal insertions
of multiples of three dots and three zeroes and vertical
insertions of the motifs "$1000$" and "$11000$", we can represent
all elements in $H_Q$. With the mappings given in the previous
step, we find a corresponding tiling for each of these elements.
On the other hand, each possible tiling can be constructed with
the small sequences of tiles given in the previous step. Thus we
find that for each possible tiling there is a corresponding
element in $H_Q$. Furthermore, we found that the number of
fermions and the number of tiles agree. So we find $N_i=t_i$, that
is, the number of elements in $H_Q$ with $i$ fermions equals the
number of tilings of the square lattice with $i$ tiles. However,
there is a small discrepancy in this one-to-one relation for the
configurations with all zeroes or all dots. For $\vec{u}=(m,-m)$
and $v_1+v_2=3p$, it is readily verified that these configurations
contain $i=[2m/3]p$ fermions. In the following we will first
compute the number of elements of $H_Q$ that these configurations
account for. We shall call this $N^{(a)}$, where $a$ stands for
anomalous. We will then compute $t^{(a)}$, the number of tilings
consisting only of the tiles that correspond to either all zeroes
or all dots (see fig. \ref{fig:zeroesdots}). Combining these
results we obtain $\Delta \equiv N^{(a)}-t^{(a)}$. Finally, since
we found a one-to-one correspondence between tilings and elements
of $H_Q$ for $i \neq [2m/3]p$, theorem \ref{tm:HQ} will then be
established with $\Delta_i$ as in equation (\ref{eq:Delta}).

As we discussed in section \ref{sec:zz} for $\vec{v}=(1,2)$, the configurations with all dots and all zeroes
actually correspond to multiple elements of the cohomology
if there is a multiple of 3 $S_1$ sites per row, that is if $\vec{u}=(3n,-3n)$. This is a direct consequence
of theorem \ref{chaincohom}, which says that a periodic chain with length $3j$ has two ground states. In
fact, these configurations account for $2^{p}$ elements each, where $p=(v_1+v_2)/3$ is the total number of
$S_1$ rows or, equivalently, of $S_2$ chains. On the other hand, for $\vec{u}=(m,-m)$ with $m \neq 3n$ they
each represent one element of the cohomology. So we find $N^{(a)}=2^{p+1}$ for $m=3n$ and $N^{(a)}=2$
otherwise.

Now, let us look at the corresponding tilings. For $\vec{u}=(m,-m)$ with $m \neq 3n$ there is no
corresponding tiling, thus there is a discrepancy of 2 between the number of tilings and the number of
elements in the cohomology. That is $\Delta \equiv N^{(a)}-t^{(a)}=2$ for $m \neq 3n$.
For $\vec{u}=(3n,-3n)$ there are tilings corresponding to the configurations with
all zeroes or all dots. Along the $\vec{u}$ direction these tilings have periodicity 3. The periodicity in
the other direction is more involved. Given the boundary condition $\vec{v}=(2r+s,r+2s)$ the tiling makes $r$
steps in the $(2,1)$ direction and $s$ steps in the $(1,2)$
direction, in arbitrary order. However, because of the periodicity of 3 in the $\vec{u}$ direction one can
also end at $(2r+s+3l,r+2s-3l)$, that is, $r+3l$ steps in the $(2,1)$ direction and $s-3l$ steps in the $(1,2)$
direction, again in arbitrary order. Thus we find
\beq
t^{(a)}=2* 3 \sum_{l\geq -r/3} \binom{r+s}{r+3l}.\nonumber
\eeq
If we define $r=3k+c$, where $c\in{0,1,2}$, we can write $t$ as:
\beq
t^{(a)}&=&6 \sum_{l=0} \binom{r+s}{c+3l}\nonumber\\
&=&6 \sum_{l=0} \Big[ \binom{r+s-2}{c+3l-2}+2\binom{r+s-2}{c+3l-1}+\binom{r+s-2}{c+3l} \Big]\nonumber\\
&=&6 \sum_{l=0} \Big[ \binom{r+s-2}{l}+\binom{r+s-2}{c+3l-1} \Big]\nonumber\\
&=&6* 2^{r+s-2}+6 \sum_{l=0} \Big[ \binom{r+s-2}{c+3l-1} \Big].\nonumber
\eeq

Repeating these steps $d$ times, such that $2d\leq r+s$, we find
\beq
t^{(a)}&=&6\sum_{l=1}^{d} 2^{r+s-2l}+6 \sum_{l=0} \binom{r+s-2d}{c+3l-d}\nonumber\\
&=&\sum_{l=0}^{2d-1} 2^{r+s-l}+6 \sum_{l=0} \binom{r+s-2d}{c+3l-d}\nonumber \\
&=&\left\{ \begin{array}{ll}
2^{r+s+1}-2 +6 \sum_{l=0} \binom{0}{c+3l-d} & \textrm{if } r+s=2d\nonumber\\
2^{r+s+1}-4 +6 \sum_{l=0} \binom{1}{c+3l-d} & \textrm{if }r+s=2d+1. \nonumber
\end{array} \right.
\eeq
For the last term we find
\beq
6 \sum_{l=0}  \binom{0}{c+3l-d} &=&\left\{ \begin{array}{ll}
6 & \textrm{if } d=3b+c\\
0 & \textrm{otherwise.}
\end{array} \right. \nonumber\\
6 \sum_{l=0} \binom{1}{c+3l-d} &=&\left\{ \begin{array}{ll}
0 & \textrm{if } d=3b+c+1\\
6 & \textrm{otherwise.}
\end{array} \right. \nonumber
\eeq

We now compare the expression for $t^{(a)}$ with the expression for the number of elements in the cohomology represented
by the configurations with all zeroes and all dots, $N^{(a)}$. For $\vec{v}=(2r+s,r+2s)$ this is
$N^{(a)}=2*2^{(v_1+v_2)/3}=2^{r+s+1}$. So we finally find
\beq
\Delta=N^{(a)}-t^{(a)}&=&\left\{ \begin{array}{llll}
-4 & \textrm{if } r+s=2d \textrm{ and } r-s=6b\\
2  & \textrm{if } r+s=2d \textrm{ and } r-s=6b\pm 2\\
4 & \textrm{if } r+s=2d+1 \textrm{ and } r-s=6b+3\\
-2  & \textrm{if } r+s=2d+1 \textrm{ and } r-s=6b\pm1. \nonumber
\end{array} \right.
\eeq
Combining this with the result $\Delta=2$ for $\vec{u}=(m,-m)$ with $m \neq 3n$, this can be cast in the
compact form of equation (\ref{eq:Delta}).

\vspace{0.5cm}
\noindent
{\bf \large{Acknowledgments}}

\vspace{0.2cm}
\noindent
We would like to thank P. Fendley for discussions and J. Jonsson for
suggested improvements of the manuscript.

\bibliographystyle{my-h-elsevier}

\end{document}